\documentclass{article}

     \usepackage[final,nonatbib]{neurips_2019}

\usepackage[utf8]{inputenc} 
\usepackage[T1]{fontenc}    
\usepackage{hyperref}       
\usepackage{url}            
\usepackage{booktabs}       
\usepackage{amsfonts}       
\usepackage{nicefrac}       
\usepackage{microtype}      

\usepackage{algorithm}
\usepackage{algorithmic}
\usepackage{cuted}
\usepackage{lipsum, color}
\usepackage{amsmath}
\usepackage{amsthm}
\usepackage{subfigure}
\usepackage{verbatim}
\usepackage{amsfonts}
\usepackage{color}
\usepackage{graphicx}
\usepackage{epstopdf}
\usepackage{graphicx}
\usepackage{filecontents}
\usepackage{mathtools}

\newtheorem{mydef}{Definition}
\newcommand{\beq}{\begin{equation}}
\newcommand{\eeq}{\end{equation}}

\DeclareMathOperator*{\argmin}{arg\,min}

\newcommand{\ra}{\rightarrow}

\newcommand{\lb}{\left(}
\newcommand{\rb}{\right)}
\newcommand{\lla}{\left\langle}
\newcommand{\rra}{\right\rangle}

\newtheorem{lemma}{Lemma}

\newcommand{\bcr}{\begin{color}{red}} 
\newcommand{\bcc}{\begin{color}{cyan}} 
\newcommand{\bcb}{\begin{color}{blue}}
\newcommand{\ec}{\end{color}}

\newcommand{\mybar}[1]{#1}
\newcommand{\myvec}[1]{#1}

\newcommand{\Cgamma}{C}
\newcommand{\Sgamma}{S^\gamma}
\newcommand{\ustar}{u^*}
\newcommand{\pstar}{p^*}
\newcommand{\thetastar}{\theta^*}
\newcommand{\Jalpha}{J^{\alpha}}

\newcommand{\RCE}{\Theta^{C}_{\mathcal{E}}}
\newcommand{\RCEp}{\Theta^{C}_{\mathcal{E}'}}
\newcommand{\RJE}{\Theta^{J^{\alpha}}_{\mathcal{E}}} 
\newcommand{\RCD}{\Theta^{C}_{\mathcal{E}'}}

\newcommand{\RC}[1]{\Theta^{C}_{#1}}

\newtheorem{theorem}{Theorem}

\title{Adaptive Smoothing Path Integral Control}

\author{%
  Dominik Thalmeier\\
  Radboud University Nijmegen\\
  Nijmegen, The Netherlands\\
  \texttt{d.thalmeier@science.ru.nl} \\
  \And
  Hilbert J. Kappen\\
  Radboud University Nijmegen\\
  Nijmegen, the Netherlands\\
  \texttt{b.kappen@science.ru.nl} \\
  \And
  Simone Totaro\\
  Universitat Pompeu Fabra\\
  Barcelona, Spain\\
  \texttt{simone.totaro@gmail.com} \\
  \And
  Vicen\c{c} G\'omez\\
  Universitat Pompeu Fabra\\
  Barcelona, Spain\\
  \texttt{vicen.gomez@upf.edu}
}

\begin{document}

\maketitle

\begin{abstract}
In Path Integral control problems a representation of an optimally controlled dynamical system can be formally computed and serve as a guidepost to learn a parametrized policy.
The Path Integral Cross-Entropy (PICE) method tries to exploit this, but is hampered by poor sample efficiency.
We propose a model-free algorithm called ASPIC (Adaptive Smoothing of Path Integral Control) that applies an inf-convolution to the cost function to speedup convergence of policy optimization.
We identify PICE as the infinite smoothing limit of such technique and show that the sample efficiency problems that PICE suffers disappear for finite levels of smoothing.
For zero smoothing this method becomes a greedy optimization of the cost, which is the standard approach in current reinforcement learning.
We show analytically and empirically that intermediate levels of smoothing are optimal, which renders the new method superior to both PICE and direct cost-optimization.
\end{abstract}

\section{Introduction}
Optimal control of non-linear dynamical systems that are continuous in time and space is hard.
Methods that have proven to work well introduce a parametrized policy like a neural network~\cite{mnih2015human,duan2016benchmarking} and directly optimize the expected cost using gradient descent~\cite{williams1992simple,peters2008reinforcement,schulman2015trust,heess2017emergence}.
To achieve a robust decrease of the expected cost, it is important to ensure that at each step the policy stays in the proximity of the old policy \cite{duan2016benchmarking}. 
This can be achieved by enforcing a trust region constraint \cite{peters2010relative,schulman2015trust} or using adaptive regularization~\cite{heess2017emergence}.
However the applicability of these methods is limited, as in each iteration of the algorithm, samples from the controlled system have to be computed.
We want to increase the convergence rate of policy optimization to reduce the number of simulations needed.

To this end we consider Path Integral control problems~\cite{kappen_prl05,KappenML2012}, that offer an alternative approach to direct cost optimization and explore if this allows to speed up policy optimization.
This class of control problems permits arbitrary non-linear dynamics and state cost, but requires a linear dependence of the control on the dynamics and a quadratic control cost~\cite{kappen_prl05,bierkens2014explicit,thijssen2015path}. 
These restrictions allow to obtain an explicit expression for the probability-density of optimally controlled system trajectories. Through this, an information-theoretical measure of the deviation of the current control policy from the optimal control can be calculated. 
The Path Integral Cross-Entropy (PICE) method~\cite{kappen2015adaptive} proposes to use this measure as a pseudo-objective for policy optimization.

In this work we analyze a new kind of smoothing technique for the cost function based on recently proposed smoothing techniques to speed up convergence in deep neural networks~\cite{chaudhari2017deep}.
We adapt this technique to Path Integral control problems and 
show that \emph{(i)}, in contrast to \cite{chaudhari2017deep}, smoothing in Path Integral control can be solved analytically, providing an expression of the gradient that can directly be computed from Monte Carlo samples and \emph{(ii)}, we can interpolate between direct cost optimization and the PICE objective. 
Remarkably, the parameter governing the smoothing can be determined independently of the number of samples.

Based on these results, we introduce the ASPIC (Adaptive Smoothing of Path Integral Control) algorithm, a model-free algorithm that uses cost smoothing to speed up policy optimization. ASPIC adjusts the smoothing parameter in each step to keep the variance of the gradient estimator at a predefined level.

\section{Path Integral Control Problems}
\label{sec:PI}
Consider the (multivariate) dynamical system
\begin{align}
\dot{\myvec{x}}_t=\myvec{f}(\myvec{x}_t,t) + g(\myvec{x}_t,t)\lb \myvec{u}(\myvec{x}_t,t) + \myvec{\xi}_t \rb,
\label{eq:dyneq}
\end{align}
with initial condition $\myvec{x}_0$.
The control policy is implemented in the control function $\myvec{u}(\myvec{x},t)$, which is additive to the white noise $\xi_t$ which has variance $\frac{\nu}{dt}$.
Given a control function $u$ and a time horizon $T$, this dynamical system induces a probability distribution $p_u(\tau)$ over state trajectories $\tau=\{\myvec{x}_t|\forall t: 0<t\leq T\}$ with initial condition $\myvec{x}_0$.

We define
the regularized expected cost
\begin{align}
\Cgamma(p_u) &=  \lla V(\tau) \rra_{p_u} + \gamma KL(p_u||p_0),
\label{eq:KLform}
\end{align}
with $V(\tau)=\int_0^T V(\myvec{x}_t,t) dt$, where the strength of the regularization $KL(p_u||p_0)$ is controlled by the parameter $\gamma$. 

The Kullback-Leibler divergence $KL(p_u||p_0)$ puts high cost to controls $u$ that bring the probability distribution $p_u$ far away from the uncontrolled dynamics
$p_0$ where $\myvec{u}(\myvec{x}_t,t)=0$.
We can also rewrite the regularizer $KL(p_u||p_0)$ directly in terms of the control function $u$ by using the Girsanov theorem, c.f.,~\cite{thijssen2015path}: $\log\frac{p_u(\tau)}{p_0(\tau)}=\frac{1}{\nu}\int_0^T \lb\frac{1}{2} \myvec{u}(\myvec{x}_t,t)^T \myvec{u}(\myvec{x}_t,t) + \myvec{u}(\myvec{x}_t,t)^T\myvec{\xi}_t \rb dt$.
The regularization then takes the form of a quadratic control cost
\begin{align*}
KL(p_u||p_0) \hspace{-.1cm}&= \hspace{-.1cm}\lla \frac{1}{\nu}\int_0^T \lb\frac{1}{2} \myvec{u}(\myvec{x}_t,t)^T \myvec{u}(\myvec{x}_t,t) + \myvec{u}(\myvec{x}_t,t)^T\myvec{\xi}_t \rb dt\rra_{p_u}\hspace{-.3cm}=\lla\frac{1}{\nu}\int_0^T \frac{1}{2} \myvec{u}(\myvec{x}_t,t)^T \myvec{u}(\myvec{x}_t,t) dt \rra_{p_u}\hspace{-.3cm},
\end{align*}
where we used that $\lla \myvec{u}(\myvec{x}_t,t)^T\myvec{\xi}_t  \rra_{p_u} = 0$. This shows that the regularization $KL(p_u||p_0)$ puts higher cost for large values of the controller $u$.

The Path Integral control problem 
is to find the optimal control function $\ustar$ that minimizes 
\begin{align}
\label{eq:ustar}
\ustar = \argmin_u \Cgamma(p_u).
\end{align}
For a more complete introduction to Path Integral control problems, see \cite{thijssen2015path,kappen2015adaptive}.

\textbf{$-$Direct cost optimization using gradient descent}:
A standard approach to find an optimal control function is to introduce a parametrized controller $\myvec{u}_\theta(\myvec{x}_t,t)$~\mbox{\cite{heess2017emergence,williams1992simple,schulman2015trust}}.
This parametrizes the path probabilities $p_{u_\theta}$ and allows to optimize the expected cost~$\Cgamma(p_{u_\theta})$ \eqref{eq:KLform} using stochastic gradient descent on the cost function:
\begin{align}
\nabla_{\theta} \Cgamma(p_{u_\theta}) &= \lla \Sgamma_{p_{u_\theta}}(\tau) \nabla_{\theta} \log p_{u_\theta}(\tau) \rra_{p_{u_\theta}},
\label{eq:reinforcegrad}
\end{align}
with the stochastic cost $\Sgamma_{p_{u_\theta}}(\tau)\coloneqq V(\tau) + \gamma \log\frac{p_{u_\theta}(\tau)}{p_0(\tau)}$ (see App.~\ref{ap:1} for details).

\textbf{$-$The Path Integral Cross-Entropy method}:
An alternative approach to direct cost-optimization was introduced in \cite{kappen2015adaptive}, and takes advantage of the analytical expression for $p_{\ustar}$, the probability density of state trajectories induced by a system with the optimal controller $\ustar$, $p_{\ustar} = \argmin_{p_u} \Cgamma(p_u)$ with $\Cgamma(p_u)$ given by Eq.~\eqref{eq:KLform}.
Finding $p_{\ustar}$ is an optimization problem over probability distributions $p_u$ that are induced
by the controlled dynamical system \eqref{eq:dyneq}.
It has been shown \cite{bierkens2014explicit,thijssen2015path} that we can solve this by replacing the minimization over $p_u$ with a  minimization over all path probability distributions $p$:
\begin{align}
p_{\ustar} \equiv \pstar &:= \argmin_{p} \Cgamma(p) = \argmin_{p} \lla V(\tau) \rra_{p} + \gamma KL(p||p_0) = \frac{1}{Z} p_0(\tau) \exp\lb -\frac{1}{\gamma}V(\tau)\rb.
\label{eq:pstar}
\end{align}
with the normalization constant $Z= \lla \exp \lb{-\frac{1}{\gamma}V(\tau)} \rb\rra_{p_0}$.
Note that the above is not a trivial statement, as we now take the minimum also over non-Markovian processes with non-Gaussian noise.

The PICE algorithm \cite{kappen2015adaptive}, instead of directly optimizing the expected cost, it minimizes the KL-divergence $KL\lb \pstar || p_{u_\theta} \rb$ which measures the deviation of a parametrized distribution $p_{u_\theta}$ from the optimal one $\pstar$.
Although direct cost optimization and PICE are different methods, their global minimum is the same if the parametrization of $u_\theta$ can express the optimal control $\ustar = u_{\thetastar}$.
The parameters $\thetastar$ of the optimal controller are found using gradient descent:
\begin{align}
\label{eq:PICEgrad}
\nabla_{\theta} KL\lb \pstar || p_{u_\theta} \rb = \frac{1}{Z_{p_{u_\theta}}}\lla \exp\lb -\frac{1}{\gamma}\Sgamma_{p_{u_\theta}}(\tau)\rb \nabla_{\theta} \log p_{u_\theta}(\tau) \rra_{p_{u_\theta}},
\end{align}
where $Z_{p_{u_\theta}} \coloneqq \lla \exp\lb -\frac{1}{\gamma} \Sgamma_{p_{u_\theta}}(\tau)\rb \rra_{p_{u_\theta}}$.

That PICE uses the optimal density as a guidepost for the policy optimization might give it an advantage compared to direct cost-optimization.
In practice however, this method only works properly if the initial guess of the controller $u_\theta$ does not deviate too much from the optimal control, as a high value of $KL\lb \pstar || p_{u_\theta} \rb$ leads to a high variance of the gradient estimator and results in bootstrapping problems of the algorithm \cite{ruiz2017particle,thalmeier2016action}.
In the next section, we introduce a method that interpolates between direct cost-optimization and the PICE method, allowing us to take advantage of the analytical optimal density without being hampered by the same bootstrapping problems as PICE.

\section{Interpolating Between Methods: Smoothing Stochastic Control Problems}
\label{sec:smthecost}
Cost function smoothing was recently introduced as a way to speed up optimization of neural networks~\cite{chaudhari2017deep}:
Optimization of a general cost function $f(\theta)$ can be speeded up by smoothing $f(\theta)$ using an inf-convolution with a distance kernel $d(\theta',\theta)$.
The smoothed function
\begin{align}
\Jalpha(\theta) = \inf_{\theta'} \alpha d(\theta',\theta) + f(\theta')
\label{eq:smoothing}
\end{align}
preserves the global minima of the function $f(\theta)$.
To apply gradient descent based optimization on $\Jalpha(\theta)$ instead of $f(\theta)$ may significantly speed up convergence \cite{chaudhari2017deep}.

We want to use this accelerative effect to find the optimal parametrization of the controller $u_\theta$.
Therefore, we smooth the cost function $\Cgamma(p_{u_\theta})$ as a function of the parameters $\theta$.
As $\Cgamma(p_{u_\theta})=\lla V(\tau) \rra_{p_{u_\theta}} + \gamma KL(p_{u_\theta}||p_0)$ is a functional on the space of probability distributions $p_{u_\theta}$, the natural ``distance''   is the KL-divergence $KL(p_{u_{\theta'}}||p_{u_\theta})$.
So we replace 
\begin{align*}
f(\theta) &\ra \Cgamma(p_{u_\theta})\\
d(\theta',\theta) &\ra KL(p_{u_{\theta'}}||p_{u_\theta}),
\end{align*}
and obtain the smoothed cost $\Jalpha(\theta)$ as
\begin{align}
\Jalpha(\theta) &= \inf_{\theta'} \alpha KL(p_{u_{\theta'}}||p_{u_\theta}) + \Cgamma(p_{u_{\theta'}})=\inf_{\theta'} \alpha KL(p_{u_{\theta'}}||p_{u_\theta}) + \gamma KL(p_{u_{\theta'}}||p_0) + \lla V(\tau) \rra_{p_{u_{\theta'}}}. 
\label{eq:smoothedcost1}
\end{align}
Note the different roles of $\alpha$ and $\gamma$: the parameter $\alpha$ penalizes the deviation of $p_{u_{\theta'}}$ from $p_{u_\theta}$, while the parameter $\gamma$ penalizes the deviation of $p_{u_{\theta'}}$ from the uncontrolled dynamics $p_0$.

\textbf{$-$ Computing the smoothed cost and its gradient}: The smoothed cost $\Jalpha$ is expressed as a minimization problem that has to be solved. Here we show that for Path Integral control problems this can be done analytically.
To do this we first show that we can replace $\inf_{\theta'} \ra \inf_{p'}$ and then solve the minimization over $p'$ analytically.
We replace the minimization over $\theta'$ by a minimization over $p'$ in two steps:
first we state an assumption that allows us to replace $\inf_{\theta'} \ra \inf_{u'}$ and then proof that for Path Integral control problems we can replace $\inf_{u'} \ra \inf_{p'}$.

We assume that for every ${u_\theta}$ and any $\alpha>0$, the minimizer $\theta^*_{\alpha,{\theta}}$ over the parameter space
\begin{align}
\theta^*_{\alpha,{\theta}} \coloneqq \argmin_{\theta'} \alpha KL(p_{u_{\theta'}}||p_{u_\theta}) + \Cgamma(p_{u_{\theta'}}) \label{eq:optimthetaalpha}
\end{align}
is the parametrization of the minimizer $u^*_{\alpha,{\theta}}$ over the function space
\begin{align*}
u^*_{\alpha,{\theta}} \coloneqq \argmin_{u'} \alpha KL(p_{u'}||p_{u_\theta}) + \Cgamma(p_{u'}),
\end{align*}
such that $u^*_{\alpha,{\theta}} \equiv u_{\theta^*_{\alpha,{\theta}}}$.
We call this assumption \emph{full parametrization}.
Naturally it is sufficient for full parametrization if $\myvec{u}_\theta(\myvec{x},t)$ is a universal function approximator with a fully observable state space $\myvec{x}$ and the time $t$ as input, although this may be difficult to achieve in practice.
With this assumption we can replace $\inf_{\theta'} \ra \inf_{u'}$. Analogously, we replace $\inf_{u'} \ra \inf_{p'}$: 
in App.~\ref{ap:5} we proof that for Path Integral control problems the minimizer $u^*_{\alpha,{\theta}}$ over the function space induces the minimizer $p^*_{\alpha,{\theta}}$ over the space of probability distributions
\begin{align}
p^*_{\alpha,{\theta}} \coloneqq &  \argmin_{p'} \alpha KL(p'||p_{u_\theta}) + \Cgamma(p'), \label{eq:smoothedcostsol}
\end{align}
such that $p^*_{\alpha,{\theta}} \equiv p_{u^*_{\alpha,{\theta}}}$. This step is similar to the derivation of
of Eq.~\eqref{eq:pstar} in Section \ref{sec:PI}, but now we have added an additional term $\alpha KL(p_{u'}||p_{u_\theta})$.

Hence, given a Path Integral control problem and a controller $u_\theta$ that satisfies full parametrization, we can replace $\inf_{\theta'} \ra \inf_{p'}$ and Eq.~\eqref{eq:smoothedcost1} becomes
\begin{align}
\Jalpha(\theta) = \inf_{p'}  \alpha KL(p'||p_{u_\theta}) + \gamma KL(p'||p_0) + \lla V(\tau) \rra_{p'}.
\label{eq:smoothedcost1a}
\end{align}
This can be solved directly:
first we compute the minimizer (see App.~\ref{ap:2})
\begin{align}
p^*_{\alpha,{\theta}}(\tau) =& \frac{1}{Z_{p_{u_\theta}}^\alpha} p_{u_\theta}(\tau) \exp \lb -\frac{1}{\gamma+\alpha} \Sgamma_{p_{u_\theta}}(\tau) \rb,
& Z_{p_{u_\theta}}^\alpha&=\lla \exp \lb -\frac{1}{\gamma+\alpha} \Sgamma_{p_{u_\theta}}(\tau) \rb \rra_{p_{u_\theta}}.
\label{eq:annealingcurve}
\end{align}
We plug this back in Eq.~\eqref{eq:smoothedcost1a} and get the smoothed cost and its gradient (see App.~\ref{ap:smoothedgrad})
\begin{align}
\Jalpha(\theta) &= -\lb \gamma+\alpha\rb \log \lla \exp \lb -\frac{1}{\gamma+\alpha} \Sgamma_{p_{u_\theta}}(\tau) \rb \rra_{p_{u_\theta}}
\label{eq:smoothedcost2}\\
\nabla_{\theta}\Jalpha(\theta) &=  -\frac{\alpha}{Z_{p_{u_\theta}}^\alpha}\lla \exp \lb -\frac{1}{\gamma+\alpha} \Sgamma_{p_{u_\theta}}(\tau) \rb \nabla_{\theta} \log p_{u_\theta}(\tau) \rra_{p_{u_\theta}}.
\label{eq:smoothedgradient}
\end{align}
Both can be estimated by samples from the distribution $p_{u_\theta}$.

\section{The ASPIC Algorithm}
\label{sec:algorithm}\label{sec:adaptivesmoothing}
In this section, we derive an iterative algorithm that takes a parametrized control function $u_\theta$ and performs smooth parameter updates starting from initial parameters $\theta_0$.
We focus on the effect that a finite $\alpha>0$ has on the iterative optimization of the control $u_\theta$ for a fixed value of $\gamma$.
For our theoretical results, we refer the reader to App.~\ref{sec:alphagamma}, where we identify several existing settings as limiting cases of the parameters $\alpha$ and $\gamma$, and to App.~\ref{sec:Theory}, where we proof that smooth updates are optimal in two-step sequential decision problems.

To simplify notation, we overload $p_{u_\theta}\ra\theta$ so that we get $\Cgamma(p_{u_\theta})\ra \Cgamma(\theta)$ and $KL(p_{u_{\theta'}}||p_{u_{\theta}}) \ra KL(\theta'||\theta)$.
We use a trust region constraint to robustly optimize the policy, c.f.,~\cite{peters2010relative, schulman2015trust, gomez2014policy}. 
We define the smoothed update with stepsize $\mathcal{E}$ as an update $\theta \ra \theta'$ 
with 
$\theta' = \RJE(\theta)$
and
\begin{align}
\RJE(\theta) \coloneqq& \argmin_{
\substack{
\theta' \\
\text{s.t.}~KL(\theta'||\theta)\leq \mathcal{E}
}
} \Jalpha(\theta'). \label{eq:smoothedupdate}
\end{align}

\textbf{$-$Smoothed and direct updates using natural gradients}: 
We first express the constraint optimization \eqref{eq:smoothedupdate} as an unconstrained optimization problem introducing a Lagrange multiplier $\beta$
\begin{align}
\label{eq:unconstraintargminf}
\theta_{n+1} &= \argmin_{\theta'} \Jalpha(\theta') + \beta KL(\theta'||\theta_{n}).
\end{align}
Following \cite{schulman2015trust} we assume that the trust region size $\mathcal{E}$ is small.
For $\mathcal{E}\ll 1$ we get  $\beta \gg 1$ and can expand $\Jalpha(\theta')$ to first and $KL(\theta'||\theta_{n})$ to second order (see App.~\ref{ap:TRPOTRPI} for the details). This gives
\begin{align}
\label{eq:natgradupdate}
\theta_{n+1}&= \theta_n - \beta^{-1}F^{-1}\left.\nabla_{\theta'}\Jalpha(\theta')\right\rvert_{\theta'=\theta_n},
\end{align}
a natural gradient update with the Fisher-matrix $F = \left.\nabla_\theta \nabla_\theta^T KL(\theta'||\theta_{n})\right\rvert_{\theta'=\theta_n}$
(we use the conjugate gradient method to approximately compute the natural gradient for high dimensional parameter spaces.
See App.~\ref{ap:4} or \cite{schulman2015trust}  for details).
Parameter $\beta$ is determined using a line search such that
\begin{align}\label{eq:eqtrust}
KL(\theta_n||\theta_{n+1})=\mathcal{E}.
\end{align}
Note that for direct updates this derivation is the same, just replace $\Jalpha$ by $\Cgamma$.

\textbf{$-$Reliable gradient estimation using adaptive smoothing}:
To compute smoothed updates using Eq.~\eqref{eq:natgradupdate} we need the gradient of the smoothed cost.
We assume full parametrization and use Eq.~\eqref{eq:smoothedgradient}, which can be estimated using $N$ weighted samples drawn 
from the distribution $p_{u_{\theta}}$:
\begin{align}
\nabla_{\theta}\Jalpha(\theta) \approx \alpha\sum_{i=1}^N w^i \log p_{u_\theta}(\tau^i). \label{eq:Jgradsamples}
\end{align}
The weights are given by
\begin{align*}
w^i & = \frac{1}{\tilde{Z}}\exp \lb -\frac{1}{\mybar{\gamma}+\alpha} \Sgamma_{p_{u_\theta}}(\tau^i) \rb,
& \tilde{Z} & =\sum_{i=1}^N \exp \lb -\frac{1}{\mybar{\gamma}+\alpha} \Sgamma_{p_{u_\theta}}(\tau^i) \rb.
\end{align*}
The variance of this estimator depends sensitively on the entropy of the weights
$H_N(w) = -\sum_{i=1}^N w^i \log(w^i)$. If the entropy is low, the total weight is concentrated on a few particles. This results in a poor gradient estimator where only a few of the particles actually contribute.
This concentration is dependent on the smoothing parameter $\alpha$: for small $\alpha$, the weights are very concentrated in a few samples, resulting in a large weight-entropy and thus a high variance of the gradient estimator.
As small $\alpha$ corresponds to strong smoothing, we want $\alpha$ to be as small as possible, but large enough to allow a reliable gradient estimation.
Therefore, we set a bound to the weight entropy $H_N(w)$.
To get a bound that is independent of the number of samples $N$, we use that in the limit of $N\ra\infty$ the weight entropy is monotonically related to the KL-Divergence $KL(p^*_{\alpha,u_\theta}||p_{u_\theta})$
\begin{align*}
KL(p^*_{\alpha,u_\theta}||p_{u_\theta})
=\lim_{N\ra \infty} \log N - H_N(w)
\end{align*}
(see App.~\ref{ap:3}).
This provides a method for choosing $\alpha$ \emph{independently of the number of samples}: we set the constraint ${KL(p^*_{\alpha,u_\theta}||p_{u_\theta})\leq\Delta}$ and determine the smallest $\alpha$ that satisfies this condition by using a line search.
Large values of $\Delta$ correspond to small values of $\alpha$ (see App.~\ref{ap:deltaalpha}) and therefore strong smoothing, we thus call parameter $\Delta$ the \emph{smoothing strength}.

\textbf{$-$A model-free algorithm}: We can compute the gradient \eqref{eq:Jgradsamples} and the KL-divergence while treating the dynamical system as a black-box. For this we write the probability distribution $p_{u_\theta}$ over trajectories $\tau$ as a Markov process $p_{u_\theta}(\tau)= \prod_{0<t<T} p_{u_\theta}(\myvec{x}_{t+dt}|\myvec{x}_{t},t)$, where the product runs over the time~$t$, which is discretized with time step $dt$.
We define the noisy action $\myvec{a}_t=\myvec{u}(\myvec{x}_t,t) + \myvec{\xi}_t$ and formulate the transitions $p_{u_\theta}(\myvec{x}_{t+dt}|\myvec{x}_{t})$ for the dynamical system \eqref{eq:dyneq} as
\begin{align*}
p_{u_\theta}(\myvec{x}_{t+dt}|\myvec{x}_{t}) &= \delta \lb\myvec{x}_{t+dt}-\mathcal{F}\lb \myvec{x}_{t},\myvec{a}_t,t \rb\rb \cdot \pi_\theta(\myvec{a}_t|t,\myvec{x}_{t}),
\end{align*}
with $\delta(\cdot)$ the Dirac delta function.
This splits the transitions up into the deterministic dynamical system $\mathcal{F}\lb \myvec{x}_{t},\myvec{a}_t,t \rb$
 and a Gaussian policy $\pi_\theta(\myvec{a}_t|t,\myvec{x}_{t})=\mathcal{N}\lb \myvec{a}_t| \myvec{u}_\theta(\myvec{x}_t,t), \frac{\nu}{dt}\rb$ with mean $\myvec{u}_\theta(\myvec{x}_t,t)$ and variance $\frac{\nu}{dt}$.
Using this we get a simplified expression for the gradient of the smoothed cost \eqref{eq:Jgradsamples} that is independent of the system dynamics, given the samples drawn from the controlled system $p_{u_\theta}$:
\begin{align*}
\nabla_{\theta}\Jalpha(\theta) \approx \alpha\sum_{i=1}^N \sum_{0<t<T} w^i \nabla_{\theta}\log \pi_\theta(\myvec{a}_t^i|t,\myvec{x}^i_{t}).
\end{align*}
Similarly we obtain an expression for the estimator of the KL divergence 
$KL(\theta_n||\theta_{n+1}) \approx \frac{1}{N}\sum_{i=1}^N \sum_{0<t<T} \log \frac{\pi_{\theta_n}(\myvec{a}_t^i|t,\myvec{x}^i_{t})}{\pi_{\theta_{n+1}}(\myvec{a}_t^i|t,\myvec{x}^i_{t})}$.
With this we formulate ASPIC (Algorithm \ref{alg:alg1}) which optimizes the parametrized policy $\pi_\theta$ by iteratively drawing samples from the controlled system. 

\section{Numerical Experiments}
\label{sec:Numerics}

\begin{figure}[t]
\vskip 0.2in
\begin{center}
\includegraphics[width=.325\columnwidth]{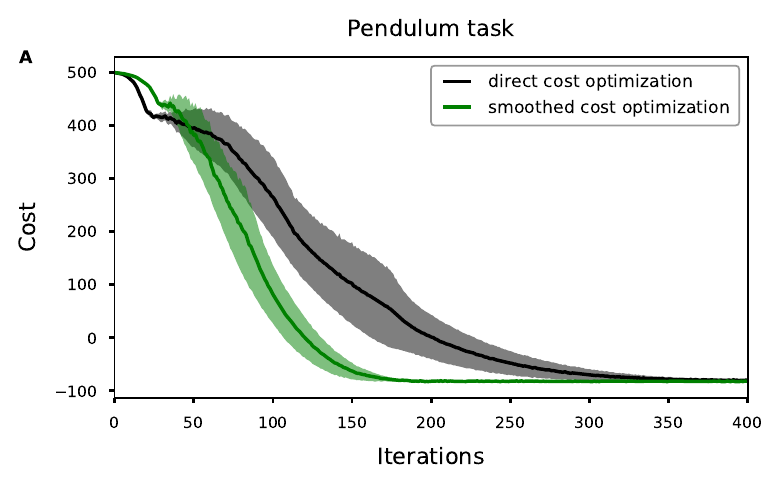} 
\includegraphics[width=.325\columnwidth]{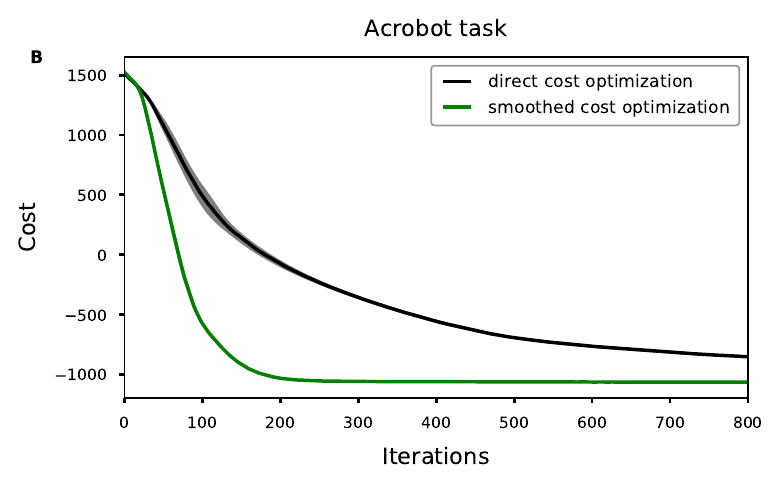}
\includegraphics[width=.325\columnwidth]{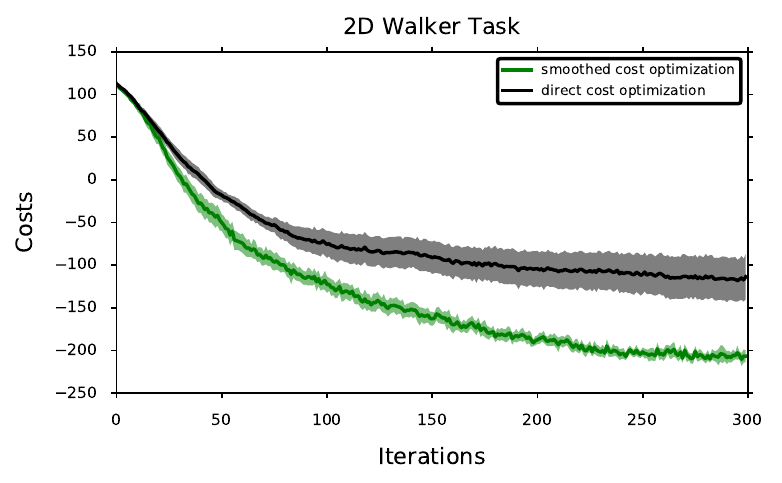}

\caption{
Smoothed cost-optimization exhibits faster convergence than direct cost-optimization in a variety of tasks.
Plots show mean and standard deviation of the cost per iteration for $10$ runs of the algorithm.
For pendulum and acrobot tasks, we used $\Delta = 0.5$ and $\mathcal{E}=0.1$ whereas for the walker, we used $\Delta=0.05\log N$ and $\mathcal{E}=0.01$. 
See App.~\ref{ap:addexp} for more details.
}
\label{fig:exp2}
\end{center}
\vskip -0.2in
\end{figure} 

We compare experimentally the convergence speed of policy optimization with and without smoothing.
For the optimization with smoothing, we use ASPIC and for the optimization without smoothing, we use a version of ASPIC where we replaced the gradient of the smoothed cost with the gradient of the cost itself. 
We consider three non-linear control problems, which violate the full parametrization assumption (pendulum swing-up task, Acrobot, and 2D walker).
The latter was simulated using OpenAI gym~\cite{1606.01540}.
For pendulum swing-up and the Acrobot tasks we used time-varying linear feedback controllers, whereas for the 2D walker task we parametrized the control $u_\theta$ using a neural network.
We provide more details about the experimental settings and additional results in App.~\ref{ap:addexp}.

\textbf{$-$Convergence rate of policy optimization}:
Fig.~\ref{fig:exp2} shows the comparison of ASPIC algorithm with smoothing against direct-cost optimization. In all three tasks, smoothing improves the convergence rate of policy optimization. Smoothed cost optimization requires less iterations to achieve the same cost reduction as direct cost-optimization, with only a negligible amount of additional computational steps that do not depend on the complexity of the simulation runs.

We can thus conclude that even in cases when the parametrized controller does not strictly meet the requirement of full parametrization needed to derive the gradient of the smoothed cost, a strong performance boost can also be achieved.

\section{Discussion}
Many policy optimization algorithms update the control policy based on a direct optimization of the cost; examples are Trust Region Policy Optimization (TRPO)~\cite{schulman2015trust} or the Path-Integral Relative Entropy Policy Search (PIREPS) \cite{gomez2014policy}, where the later is particularly developed for Path Integral control problems.
The main novelty of this work is the application of the idea of smoothing as introduced in \cite{chaudhari2017deep} to Path Integral control problems.
This allows to outperform direct cost-optimization and achieve faster convergence rates with only a negligible amount of computational overhead.

This procedure bears similarities to an adaptive annealing scheme, with the smoothing parameter playing the role of an artificial temperature.
However in contrast to classical annealing schemes, such as simulated annealing, changing the smoothing parameter does not change the optimization target: the minimum of the smoothed cost remains the optimal control solution for all levels of smoothing.

In the weak smoothing limit, ASPIC directly optimizes the cost using trust region constrained updates, similar to the TRPO algorithm \cite{schulman2015trust}.
TRPO differs from ASPIC's weak smoothing limit by additionally using certain variance reduction techniques for the gradient estimator: They replace the stochastic cost in the gradient estimator by the easier-to-estimate advantage
function, which has a state dependent baseline and only takes into account future expected cost. 
Since this depends on the linearity of the gradient in the stochastic cost and this dependence is non-linear for the gradient of the smoothed cost, 
we cannot directly incorporate these variance reduction techniques in ASPIC.

In the strong smoothing limit ASPIC becomes a version of PICE \cite{kappen2015adaptive} that---unlike the plain PICE algorithm---uses a trust region constraint to achieve robust updates. 
The gradient estimation problem that appears in the PICE algorithm was previously addressed in \cite{ruiz2017particle}: they proposed a heuristic that allows to reduce the variance of the gradient estimator by adjusting the particle weights used to compute the policy gradient. 
In \cite{ruiz2017particle} this heuristic is introduced as an ad hoc fix of the sampling problem and the adjustment of the weights introduces a bias with possible unknown side effects.
Our study sheds a new light on this, as adjusting the particle weights corresponds to a change of the smoothing parameter in our case.

\subsubsection*{Acknowledgments}

The research leading to these results has received funding from ``La Caixa'' Banking Foundation. Vicen\c{c} G\'omez is supported by the Ramon y Cajal program RYC-2015-18878 (AEI/MINEICO/FSE,UE) and the Mar\'ia de Maeztu Units of Excellence Programme (MDM-2015-0502).

\bibliography{mybib}{}
\bibliographystyle{abbrv}

\begin{appendix}

\section{Derivation of the policy gradient}
\label{ap:1}
Here we derive Eq.~\eqref{eq:reinforcegrad}. We write  $\Cgamma(p_{u_\theta})=\lla \Sgamma_{u_\theta}(\tau) \rra_{p_{u_\theta}}$, with $\Sgamma_{u_\theta}(\tau)\coloneqq V(\tau) + \gamma \log\frac{p_{u_\theta}(\tau)}{p_0(\tau)}$ and take the derivative of Eq.~\eqref{eq:KLform}:
\begin{align}
\nabla_{\theta}\lla \Sgamma_{u_\theta}(\tau) \rra_{p_{u_\theta}} &=
 \nabla_{\theta}\lla V(\tau) + \gamma \log\frac{p_{u_\theta}(\tau)}{p_0(\tau)} \rra_{p_{u_\theta}}
\end{align}
Now we introduce the importance sampler $p_{u_{\theta'}}$ and correct for it. 
\begin{align}
\nabla_{\theta}\lla \Sgamma_{u_\theta}(\tau) \rra_{p_{u_\theta}}&= \nabla_{\theta}\lla \frac{p_{u_{\theta}}(\tau)}{p_{u_{\theta'}}(\tau)} \lb V(\tau) + \gamma \log\frac{p_{u_\theta}(\tau)}{p_0(\tau)}\rb \rra_{p_{u_{\theta'}}}
\end{align}
This is true for all $\theta'$ as long as $p_{u_{\theta}}(\tau)$ and $p_{u_{\theta'}}(\tau)$ are absolutely continuous to each other.
Taking the derivative we get:
\begin{align}
&= \lla \frac{\nabla_{\theta} p_{u_{\theta}}(\tau)}{p_{u_{\theta'}}(\tau)} \lb V(\tau) + \gamma \log\frac{p_{u_\theta}(\tau)}{p_0(\tau)}\rb \rra_{p_{u_{\theta'}}} + \lla \frac{p_{u_{\theta}}(\tau)}{p_{u_{\theta'}}(\tau)} \lb \gamma \frac{1}{p_{u_\theta}(\tau)}\nabla_{\theta}p_{u_\theta}(\tau)\rb \rra_{p_{u_{\theta'}}} \\
&= \lla \lb\nabla_{\theta} \log p_{u_{\theta}}(\tau)\rb \lb V(\tau) + \gamma \log\frac{p_{u_\theta}(\tau)}{p_0(\tau)}\rb \rra_{p_{u_{\theta}}} + \gamma \nabla_{\theta}\lla \frac{1}{p_{u_{\theta'}}(\tau)}  p_{u_\theta}(\tau) \rra_{p_{u_{\theta'}}} \\
&= \lla  \Sgamma_{u_\theta}(\tau) \nabla_{\theta} \log p_{u_{\theta}}(\tau) \rra_{p_{u_{\theta}}} + \gamma \nabla_{\theta}\lla 1 \rra_{p_{u_{\theta}}} \\
&= \lla  \Sgamma_{u_\theta}(\tau) \nabla_{\theta} \log p_{u_{\theta}}(\tau) \rra_{p_{u_{\theta}}}.
\end{align}

\section{Smoothing Stochastic Control Problems}
\label{app:smoothing}

\subsection{Replacing Minimization over $u$ by Minimization over $p'$}
\label{ap:5}
Here we show that for
\begin{align}
\Jalpha(\theta) = \inf_{u'}  \alpha KL(p_{u'}||p_{u_\theta}) + \gamma KL(p_{u'}||p_0) + \lla V(\tau) \rra_{p'}
\label{apeq:jminu}
\end{align}
we can replace the minimization over $u$ by a minimization over $p'$ to obtain Eq.~\eqref{eq:smoothedcost1a}.
For this, we need to show that the minimizer $p_{\alpha,\theta}^*$ of Eq.~\eqref{eq:smoothedcost1a} is induced by $u_{\alpha,\theta}^*$, the minimizer of Eq.~\eqref{apeq:jminu}:
\begin{align*}
p_{\alpha,\theta}^* \equiv p_{u_{\alpha,\theta}^*}.
\end{align*}
The solution to \eqref{eq:smoothedcost1a} is given by (see App.~\ref{ap:2})
\begin{align}
p^*_{\alpha,\theta}&=\frac{1}{Z} p_{u_\theta}(\tau) \exp \lb -\frac{1}{\gamma+\alpha} \Sgamma_{p_{u_\theta}}(\tau) \rb = \frac{1}{Z} p_{u_\theta}(\tau)\lb\frac{p_0(\tau)}{p_{u_\theta}(\tau)}\rb^\frac{\gamma}{\gamma+\alpha} \exp \lb -\frac{1}{\gamma+\alpha} V(\tau) \rb.
\end{align}
We rewrite
\begin{align*}
p_{0}(\tau)\lb\frac{p_{u_\theta}(\tau)}{p_0(\tau)}\rb^{1-\frac{\gamma}{\gamma+\alpha}} &= p_0(\tau)\exp\lb \lb1-\frac{\gamma}{\gamma+\alpha}\rb\int_0^T \lb\frac{1}{2} \myvec{u}_\theta(\myvec{x}_t,t)^T \myvec{u}_\theta(\myvec{x}_t,t) + \myvec{u}_\theta(\myvec{x}_t,t)^T\myvec{\xi}_t \rb dt\rb,
\end{align*}
where we used the Girsanov theorem~\cite{bierkens2014explicit,thijssen2015path} (and set $\nu=1$ for simpler notation).
With $\myvec{\tilde{u}}_\theta(\myvec{x}_t,t)\coloneqq \lb 1-\frac{\gamma}{\gamma+\alpha} \rb  \myvec{u}_\theta(\myvec{x}_t,t)$ this gives
\begin{align*}
p_{0}(\tau)\lb\frac{p_{u_\theta}(\tau)}{p_0(\tau)}\rb^{1-\frac{\gamma}{\gamma+\alpha}}
&= p_0(\tau)\exp\lb \int_0^T \lb\frac{1}{2} \myvec{\tilde{u}}_\theta(\myvec{x}_t,t)^T \myvec{\tilde{u}}_\theta(\myvec{x}_t,t) + \myvec{\tilde{u}}_\theta(\myvec{x}_t,t)^T\myvec{\xi}_t \rb dt\rb
\cdot \\ &\quad \cdot
\exp\lb \int_0^T \lb\frac{1}{2} \frac{\gamma}{ \alpha}\myvec{\tilde{u}}_\theta(\myvec{x}_t,t)^T \myvec{\tilde{u}}_\theta(\myvec{x}_t,t)  \rb dt\rb \\
&= 
p_{{\tilde{u}}_\theta}(\tau)
\exp\lb \int_0^T \lb\frac{1}{2} \frac{\gamma}{ \alpha}\myvec{\tilde{u}}_\theta(\myvec{x}_t,t)^T \myvec{\tilde{u}}_\theta(\myvec{x}_t,t) \rb dt \rb.
\end{align*}
So we get
\begin{align}
p^*_{\alpha,\theta}
&= \frac{1}{Z} p_{{\tilde{u}}_\theta}(\tau)
\exp\lb \int_0^T \lb\frac{1}{2} \frac{\gamma}{ \alpha}\myvec{\tilde{u}}_\theta(\myvec{x}_t,t)^T \myvec{\tilde{u}}_\theta(\myvec{x}_t,t) \rb dt \rb \exp \lb -\frac{1}{\gamma+\alpha} V(\tau) \rb.
\end{align}
This has the form of an optimally controlled distribution with dynamics 
\begin{align}
\label{eq:appendixBdynsys}
\dot{\myvec{x}}_t=\myvec{f}(\myvec{x}_t,t) + g(\myvec{x}_t,t)\left( \myvec{\tilde{u}}_\theta(\myvec{x}_t,t)+\myvec{\hat{u}}(\myvec{x}_t,t) + \myvec{\xi}_t \right)
\end{align}
and cost
\begin{align}
\lla \int_0^T \frac{1}{\gamma+\alpha}V(\myvec{x}_t,t) -\frac{1}{2} \frac{\gamma}{ \alpha}\myvec{\tilde{u}}_\theta(\myvec{x}_t,t)^T \myvec{\tilde{u}}_\theta(\myvec{x}_t,t)  dt+\int_0^T \lb\frac{1}{2} \myvec{\hat{u}}(\myvec{x}_t,t)^T \myvec{\hat{u}}(\myvec{x}_t,t) + \myvec{\hat{u}}(\myvec{x}_t,t)^T\myvec{\xi}_t\rb dt \rra_{p_{\hat{u}}}.
\end{align}
This is a Path Integral control problem with state cost $\int_0^T \frac{1}{\gamma+\alpha}V(\myvec{x}_t,t)- \frac{1}{2} \frac{\gamma}{ \alpha}\myvec{\tilde{u}}_\theta(\myvec{x}_t,t)^T \myvec{\tilde{u}}_\theta(\myvec{x}_t,t)  dt$ which is well defined with $\myvec{\tilde{u}}_\theta(\myvec{x}_t,t)=\lb 1-\frac{\gamma}{\gamma+\alpha} \rb  \myvec{u}_\theta(\myvec{x}_t,t)$.

Let $\hat{u}^*$ be the optimal control of this Path Integral control problem. Then $p^*_{\alpha,\theta}$ is induced by Eq.~ \eqref{eq:appendixBdynsys} with $\hat{u} = \hat{u}^*$. This is equivalent to say that $p^*_{\alpha,\theta}$ is induced by Eq.~ \eqref{eq:dyneq}. As $p^*_{\alpha,\theta}$ is the density that minimizes Eq.~\eqref{eq:smoothedcost1a}, $\tilde{u}_\theta + \hat{u}^*$ is minimizing Eq.~\eqref{apeq:jminu}.

\subsection{Minimizer of smoothed cost}
\label{ap:2}
Here we want to proof Eq.~\eqref{eq:annealingcurve}:
\begin{align}
p^*_{\alpha,\theta}(\tau) \coloneqq &  \argmin_{p'} \alpha KL(p'||p_{u_\theta}) + \lla \Sgamma_{p_{u_\theta}}(\tau) \rra_{p'} \\
=&  \argmin_{p'} \lla \alpha \log\frac{p'(\tau)}{p_{u_\theta}(\tau)} +   V(\tau) + \gamma \log\frac{p'(\tau)}{p_0(\tau)} \rra_{p'}.
\end{align}
For this we take the variational derivative and set it to zero:
\begin{align}
0 &= \left. \frac{\delta}{\delta p'(\tau)}\lla \alpha \log\frac{p'(\tau)}{p_{u_\theta}(\tau)} +   V(\tau) + \gamma \log\frac{p'(\tau)}{p_0(\tau)} +\kappa \rra_{p'} \right\rvert_{p' = p^*_{\alpha,\theta}},
\end{align}
where we added a Lagrange multiplier $\kappa$ to ensure normalization.
We get
\begin{align}
0&= \left. \alpha \log\frac{p'(\tau)}{p_{u_\theta}(\tau)} +   V(\tau) + \gamma \log\frac{p'(\tau)}{p_0(\tau)} +\kappa  \right\rvert_{p' = p^*_{\alpha,\theta}},
\end{align}
from which follows
\begin{align}
p^*_{\alpha,\theta}(\tau) =& \exp\lb \frac{\kappa}{\alpha +\gamma} \rb p_{u_\theta}(\tau)^{\frac{\alpha}{\alpha+\gamma}} p_{0}(\tau)^{\frac{\gamma}{\alpha+\gamma}} \exp \lb -\frac{1}{\gamma+\alpha} V(\tau) \rb \\
&= \exp\lb \frac{\kappa}{\alpha +\gamma} \rb  p_{u_\theta}(\tau) \exp \lb -\frac{1}{\gamma+\alpha} V(\tau) - \frac{\gamma}{\alpha+\gamma} \log \frac{p_{u_\theta}(\tau)}{p_{0}(\tau)}\rb \\
&= \exp\lb \frac{\kappa}{\alpha +\gamma} \rb p_{u_\theta}(\tau) \exp \lb -\frac{1}{\gamma+\alpha} \Sgamma_{p_{u_\theta}}(\tau)\rb,
\end{align}
where $\kappa$ is chosen such that the distribution is normalized.

\subsection{Derivation of the gradient of the smoothed cost function}
\label{ap:smoothedgrad}
Here we derive Eq.~\eqref{eq:smoothedgradient} by taking the derivative of Eq.~\eqref{eq:smoothedcost2}:

\begin{align}
\nabla_{\theta}\Jalpha(\theta) &=
 -\lb \gamma+\alpha\rb \nabla_{\theta} \log \lla \exp \lb -\frac{1}{\gamma + \alpha}\lb V(\tau) + \gamma \log\frac{p_{u_\theta}(\tau)}{p_0(\tau)} \rb\rb \rra_{p_{u_\theta}} \\
&=  -\frac{\gamma + \alpha}{Z_{p_{u_\theta}}^\alpha} \nabla_{\theta} \lla \exp \lb- \frac{1}{\gamma + \alpha}\lb V(\tau) + \gamma \log\frac{p_{u_\theta}(\tau)}{p_0(\tau)} \rb\rb\rra_{p_{u_\theta}}.
\end{align}
Now we introduce the importance sampler $p_{u_{\theta'}}$ and correct for it. 
\begin{align}
\nabla_{\theta}\Jalpha(\theta)
&= -\frac{\gamma + \alpha}{Z_{p_{u_\theta}}^\alpha}  \nabla_{\theta}  \lla \frac{p_{u_{\theta}}(\tau)}{p_{u_{\theta'}}(\tau)} \exp\lb -\frac{1}{\gamma + \alpha}\lb V(\tau) + \gamma \log\frac{p_{u_\theta}(\tau)}{p_0(\tau)} \rb\rb \rra_{p_{u_{\theta'}}} \\
&= -\frac{\gamma + \alpha}{Z_{p_{u_\theta}}^\alpha}  \nabla_{\theta}  \lla \frac{p_{0}(\tau)^\frac{\gamma}{\gamma+\alpha}}{p_{u_{\theta'}}(\tau)}\lb p_{u_{\theta}}(\tau)\rb^{\frac{\alpha}{\gamma + \alpha}}
\exp\lb -\frac{1}{\gamma + \alpha} V(\tau)\rb \rra_{p_{u_{\theta'}}} \\
&= -\frac{\alpha}{Z_{p_{u_\theta}}^\alpha}  \lla \frac{1}{p_{u_{\theta'}}(\tau)}\lb \frac{p_{u_{\theta}}(\tau)}{p_{0}(\tau)}\rb^{-\frac{\gamma}{\gamma + \alpha}}
\exp\lb -\frac{1}{\gamma + \alpha} V(\tau)\rb  \nabla_{\theta} p_{u_{\theta}} \rra_{p_{u_{\theta'}}} \\
&=  -\frac{\alpha}{Z_{p_{u_\theta}}^\alpha}\lla \exp \lb -\frac{1}{\gamma+\alpha} \Sgamma_{p_{u_\theta}}(\tau) \rb \nabla_{\theta} \log p_{u_\theta}(\tau) \rra_{p_{u_\theta}}.
\end{align}

\section{PICE, Direct Cost-Optimization and Risk Sensitivity as Limiting Cases of Smoothed Cost Optimization}
\label{sec:alphagamma}
The smoothed cost and its gradient depend on the two parameters $\alpha$ and $\gamma$, which come from the smoothing Eq.~\eqref{eq:smoothing} and the definition of the control problem \eqref{eq:KLform} respectively.
Although at first glance the two parameters seem to play a similar role, they change different properties of the smoothed cost $\Jalpha(\theta)$ when they are varied.

In the expression for the smoothed cost \eqref{eq:smoothedcost2}, the parameter $\alpha$ only appears in the sum $\gamma + \alpha$.
Varying it changes the effect of the smoothing but leaves the optimum $\thetastar = \argmin_\theta \Jalpha(\theta)$ of the smoothed cost invariant. 
Here we show that smoothing leaves the global optimum of the cost $\Cgamma(p_{u_{\theta}})$ invariant.
As $KL(p_{{u_{\theta'}}}||p_{{u_{\theta}}})\geq 0$ we have that 
\begin{align*}
\Jalpha(\theta) = \inf_{\theta'} \Cgamma(p_{{u_{\theta'}}}) + \alpha KL(p_{{u_{\theta'}}}||p_{{u_{\theta}}}) \geq \inf_{\theta'} \Cgamma(p_{{u_{\theta'}}}) = \Cgamma(p_{{u_{\theta^*}}}).
\end{align*}
To show that the global minimum $\theta^*$ of $\Cgamma$ is also the global minimum of $\Jalpha$, it is thus sufficient to show that 
\begin{align*}
\Jalpha(\theta^*) \leq \Cgamma(p_{{u_{\theta^*}}}).
\end{align*}
We have
\begin{align*}
\Jalpha(\theta^*) = \inf_{\theta'} \Cgamma(p_{{u_{\theta'}}}) + \alpha KL(p_{{u_{\theta'}}}||p_{{u_{\theta^*}}}).
\end{align*}
Using that the minimum of a sum of terms is never larger than the sum of the minimum of terms, we get
\begin{align*}
\Jalpha(\theta^*) &\leq \lb\inf_{\theta'} \Cgamma(p_{{u_{\theta'}}}) \rb + \lb\inf_{\theta'} \alpha KL(p_{{u_{\theta'}}}||p_{{u_{\theta^*}}}) \rb \\
&=  \Cgamma(p_{{u_{\theta^*}}})  + \alpha KL(p_{{u_{\theta^*}}}||p_{{u_{\theta^*}}}) \\
&= \Cgamma(p_{{u_{\theta^*}}}).
\end{align*}

We also expect local maxima to be also preserved for large-enough smoothing parameter $\alpha$.
This would correspond to small time smoothing by the associated Hamilton-Jacobi partial differential equation~\cite{chaudhari2017deep}.

We therefore call $\alpha$ the \emph{smoothing parameter}.
The larger $\alpha$, the weaker the smoothing; in the limiting case $\alpha \ra \infty$, smoothing is turned off as we can see from Eq.~\eqref{eq:smoothedcost2}: for very large $\alpha$, the exponential and the logarithmic function linearise, $\Jalpha(\theta) \ra \Cgamma(p_{u_\theta})$ and we recover direct cost-optimization.
For the limiting case $\alpha \ra 0$, we recover the PICE method: the optimizer $p^*_{\alpha,\theta}$ becomes equal to the optimal density $\pstar$ and the gradient on the smoothed cost \eqref{eq:smoothedgradient} becomes proportional to the PICE gradient \eqref{eq:PICEgrad}:
\begin{align*}
\lim_{\alpha \ra 0} \frac{1}{\alpha} \nabla_{\theta}\Jalpha(\theta) = \nabla_{\theta} KL(\pstar||p_{u_\theta}).
\end{align*}

Varying $\gamma$ changes the control problem and thus its optimal solution. 
For $\gamma \ra 0$, the control cost becomes zero. In this case the cost only consists of the state cost and arbitrary large controls are allowed.
We get
\begin{align*}
\Jalpha(\theta) &= - \alpha \log \lla \exp \lb -\frac{1}{\alpha} V(\tau) \rb \rra_{p_{u_\theta}}.
\end{align*}
This expression is identical to the risk sensitive control cost proposed in \cite{fleming2002risk,fleming1995risk,broek2012risk}. Thus, for $\gamma=0$, the smoothing parameter $\alpha$ controls the risk-sensitivity, resulting in risk seeking objectives for $\alpha>0$ and risk avoiding objectives for $\alpha<0$.  
In the limiting case $\gamma \ra \infty$, the problem becomes trivial; the optimal controlled dynamics becomes equal to the uncontrolled dynamics: $\pstar \ra p_0$, c.f., Eq.~\eqref{eq:pstar}, and $\ustar \ra 0$.

If both parameters $\alpha$ and $\gamma$ are small, the problem is hard (see \cite{ruiz2017particle,thalmeier2016action}) as many samples are needed to estimate the smoothed cost.
The problem becomes feasible if either $\alpha$ or $\gamma$ is increased.
Increasing $\gamma$ however, changes the control problem, while increasing $\alpha$ weakens the effect of smoothing. 

\section{The effect of cost function smoothing on policy optimization}
\label{sec:Theory}

\begin{figure}[t]
\vskip 0.2in
\begin{center}
\centerline{\includegraphics[width=.9\columnwidth]{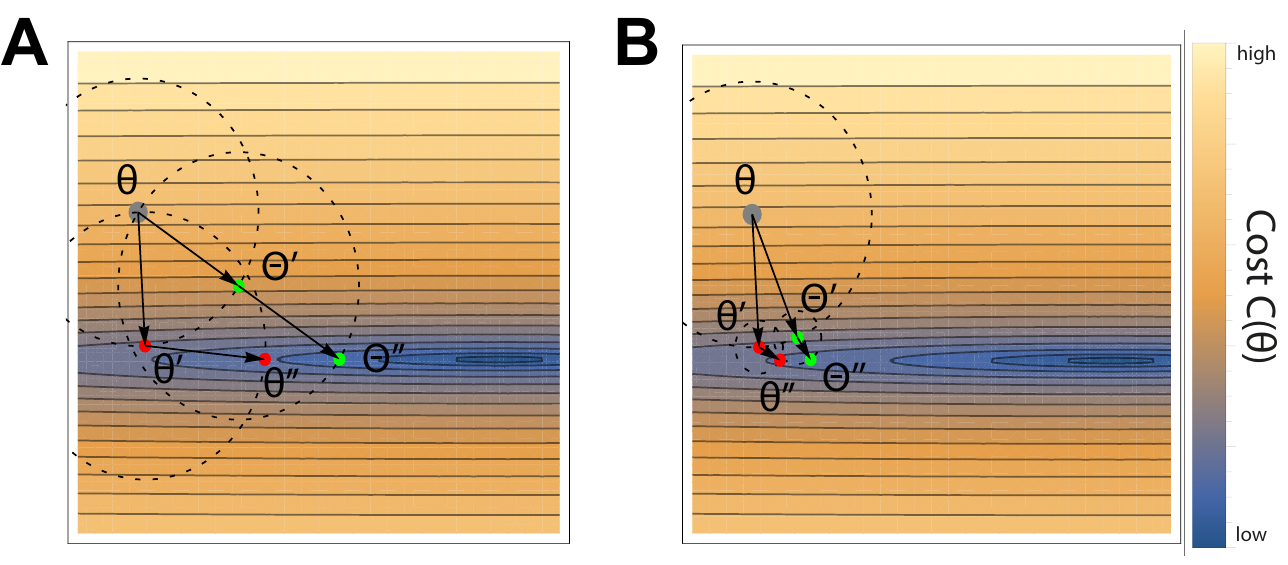}}
\caption{
Illustration of optimal two-step updates compared with two consecutive direct updates.
Illustrated is a two-dimensional cost landscape $\Cgamma(\theta)$ parametrized by $\theta$.
Dark colors represent low cost, while light colors represent high cost.
Green dots indicate the optimal two-step update $\theta \ra \Theta' \ra \Theta''$
while 
red dots indicates two consecutive direct updates $\theta \ra \theta' \ra \theta''$ with $\theta'=\RCE(\theta)$ and $\theta''=\RCD(\theta')$.
The dashed circles indicate trust regions.
$\theta'$, $\theta''$ and $\Theta''$ are the minimizers of the cost in the trust regions around $\theta$, $\theta'$ and $\Theta'$ respectively.
$\Theta'$ is chosen such that the cost $\Cgamma(\Theta'')$ after the subsequent direct update is minimized.
In both panels, the final cost after an optimal two-step update $\Cgamma(\Theta'')$ is smaller than the final cost after two direct updates $\Cgamma(\theta'')$.
(A) Equal sizes of the update steps, $\mathcal{E}=\mathcal{E}'$.
(B) When the size of the second step becomes small $\mathcal{E'}\ll\mathcal{E}$, the smoothed update $\theta \ra \Theta'$ becomes more similar to the direct update $\theta \ra \theta'$.
}
\label{fig0}
\end{center}
\vskip -0.2in
\end{figure}

We introduced smoothing as a way to speed up policy optimization compared to a direct optimization of the cost.
In this section we analyse policy optimization with and without smoothing and show analytically how smoothing can speed up policy optimization.
To simplify notation, we overload $p_{u_\theta}\ra\theta$ so that we get $\Cgamma(p_{u_\theta})\ra \Cgamma(\theta)$ and $KL(p_{u_{\theta'}}||p_{u_{\theta}}) \ra KL(\theta'||\theta)$.

We use a trust region constraint to robustly optimize the policy, c.f.,~\cite{peters2010relative, schulman2015trust, gomez2014policy}. 
There are two options. On the one hand, we can directly optimize the cost $\Cgamma$: 
\begin{mydef}
We define the direct update with stepsize $\mathcal{E}$ as an update $\theta \ra \theta'$ 
with 
$\theta' = \RCE(\theta)$
and
\begin{align}
\RCE(\theta) \coloneqq& \argmin_{
\substack{
\theta' \\
\text{s.t.}~KL(\theta'||\theta)\leq \mathcal{E}
}
} \Cgamma(\theta').
\end{align}
The direct update results in the minimal cost that can be achieved after one single update.
We define the optimal one-step cost
\begin{align*}
C^*_{\mathcal{E}}(\theta) \coloneqq&\min_{
\substack{
\theta' \\
\text{s.t.}~KL(\theta'||\theta)\leq \mathcal{E}
}
} \Cgamma(\theta').
\end{align*}
\end{mydef}
On the other hand we can optimize the smoothed cost $J^\alpha$:
\begin{mydef}
We define the smoothed update with stepsize $\mathcal{E}$ as an update $\theta \ra \theta'$ 
with 
$\theta' = \RJE(\theta)$
and
\begin{align}
\RJE(\theta) \coloneqq& \argmin_{
\substack{
\theta' \\
\text{s.t.}~KL(\theta'||\theta)\leq \mathcal{E}
}
} \Jalpha(\theta'). \label{eq:smoothedupdate2}
\end{align}
\end{mydef}
While a direct update achieves the minimal cost that can be achieved after a single update, we show below that a smoothed update can result in a faster cost reduction if more than one update step is performed.
\begin{mydef}
\label{def:twostep}
We define the optimal two-step update $\theta \ra \Theta' \ra \Theta''$ as an update that results in the lowest cost that can be achieved with a two-step update $\theta \ra \theta' \ra \theta''$ with fixed stepsizes $\mathcal{E}$ and $\mathcal{E}'$ respectively:
\begin{align*}
\Theta', \Theta'' \coloneqq&\argmin_{
\substack{
\theta',\theta'' \\
\text{  s.t.}~KL(\theta''||\theta')\leq \mathcal{E'}\\
\text{      }~KL(\theta'||\theta)\leq \mathcal{E}
}
} \Cgamma(\theta'')
\end{align*}
and the corresponding optimal two-step cost
\begin{align}
C^*_{\mathcal{E},\mathcal{E}'}(\theta) \coloneqq&\min_{
\substack{
\theta' \\
\text{s.t.}~KL(\theta'||\theta)\leq \mathcal{E}
}
}
\min_{
\substack{
\theta''\\
\text{  s.t.}~KL(\theta''||\theta')\leq \mathcal{E'}
}
}
\Cgamma(\theta'')
= &\min_{
\substack{
\theta' \\
\text{s.t.}~KL(\theta'||\theta)\leq \mathcal{E}
}
}
\Cgamma\lb \RCEp(\theta') \rb. \label{eq:ft}
\end{align}
\end{mydef}
In Fig.~\ref{fig0} we illustrate how such an optimal two-step update leads to a faster decrease of the cost than two consecutive direct updates.

\begin{theorem}
\label{th:smooth}
\emph{Statement 1:}
For all $\mathcal{E}$, $\alpha$ there exists an $\mathcal{E}'$, such that a smoothed update with stepsize $\mathcal{E}$ followed by a direct update with stepsize $\mathcal{E}'$ is an optimal two-step update: 
\begin{align*}
\Theta'&=\RJE(\theta) \\
\Theta''&=\RCEp(\Theta')\\
& \\
\Rightarrow \Cgamma\lb \Theta'' \rb &= C^*_{\mathcal{E},\mathcal{E}'}(\theta)
\end{align*}
The size of the second step $\mathcal{E}'$ is a function of $\theta$ and $\alpha$. 

\emph{Statement 2:}
$\mathcal{E}'$ is monotonically decreasing in $\alpha$.
\end{theorem}
While it is evident from Eq.~\eqref{eq:ft} that the second step of the optimal two-step update must be a direct update, the statement that the first step is a smoothed update is non-trivial.

We split the proof into three subsections: in the first subsection, we state and proof a lemma that we need to proof statement 1. In the second subsection, we proof statement 1 and in the third subsection, we proof statement 2.

\subsection{Lemma}
\begin{lemma}
\label{th:lemma1}
With $\theta^*_{\alpha,\theta}$ defined as in Eq.~\eqref{eq:optimthetaalpha} and $\mathcal{E}_\alpha(\theta) = KL\lb\theta^*_{\alpha,\theta}||\theta\rb$ we can rewrite $\Jalpha(\theta)$:
\begin{align}
\label{eq:lemma1}
\Jalpha(\theta) &= \left. \Cgamma\lb \RCD(\theta)\rb \right\rvert_{{\mathcal{E}'}=\mathcal{E}_\alpha(\theta)} + \alpha\mathcal{E}_\alpha(\theta).
\end{align}
\end{lemma}

\begin{proof}
With the definition of $\theta^*_{\alpha,\theta}$ as the minimizer of $\Cgamma(\theta')+\alpha KL(\theta'||\theta)$ (see \eqref{eq:optimthetaalpha}) we have
\begin{align*}
\Jalpha(\theta)&=\Cgamma\lb \theta^*_{\alpha,\theta} \rb + \alpha KL(\theta^*_{\alpha,\theta}||\theta)\\
&=\Cgamma\lb \theta^*_{\alpha,\theta} \rb + \alpha\mathcal{E}_\alpha(\theta).
\end{align*}
What is left to show is that 
\begin{align*}
\theta^*_{\alpha,\theta} \equiv \RC{\mathcal{E}_\alpha(\theta)} (\theta).
\end{align*}
As $\RC{\mathcal{E}_\alpha(\theta)} (\theta)$ is the minimizer of the cost $\Cgamma$ within the trust region defined by \mbox{$\{\theta': KL(\theta'||\theta)\leq \mathcal{E}_\alpha(\theta)\}$} we have to show that 
\begin{enumerate}
\item $\theta^*_{\alpha,\theta}$ lies within this trust region, 
\item $\Cgamma\lb \theta^*_{\alpha,\theta} \rb$ is a minimizer of the cost $\Cgamma$ within this trust region.
\end{enumerate}
The first point is trivially true as $KL(\theta^*_{\alpha,\theta}||\theta)= \mathcal{E}_\alpha(\theta)$ by definition. Hence $\theta^*_{\alpha,\theta}$
lies at the boundary of this trust region and therefore in it, as the boundary belongs to the trust region.
The second point we proof by contradiction:
Given $\theta^*_{\alpha,\theta}$ is not minimizing the cost within the trust region,
then there exists a $\hat{\theta}$ with \mbox{$\Cgamma(\hat{\theta})<\Cgamma(\theta^*_{\alpha,\theta})$} and $KL(\hat{\theta}||\theta)\leq \mathcal{E}_\alpha(\theta)= KL(\theta^*_{\alpha,\theta}||\theta)$.
Therefore it must hold that
\begin{align*}
\Cgamma(\hat{\theta})+\alpha KL(\hat{\theta}||\theta)<\Cgamma(\theta^*_{\alpha,\theta}) +\alpha KL(\theta^*_{\alpha,\theta},\theta)
\end{align*}
which is a contradiction, as $\theta^*_{\alpha,\theta}$ is the minimizer of $\Cgamma(\theta')+\alpha KL(\theta'||\theta)$.
\end{proof}

\subsection{Proof of Statement 1}
Here we show that for every $\alpha$ and $\theta$ there exists an ${\mathcal{E}'}=\mathcal{E}^*_\alpha(\theta)$ such that 
\begin{align}
\label{eq:smoothingeq}
\left. C \lb \RCD\lb \RJE(\theta) \rb \rb \right\rvert_{{\mathcal{E}'}=\mathcal{E}^*_\alpha(\theta)} = \left. C^*_{\mathcal{E},{\mathcal{E}'}} \right\rvert_{{\mathcal{E}'}=\mathcal{E}^*_\alpha(\theta)}.
\end{align}
\begin{proof}
As $\Jalpha(\theta)$ is the infimum of  $\Cgamma(\theta')+\alpha KL(\theta'||\theta)$, we have for any ${\mathcal{E}'}>0$
\begin{align*}
\Jalpha(\theta) &\leq  \Cgamma\lb \RCD(\theta)\rb+\alpha KL\lb \RCD(\theta)||\theta\rb .
\end{align*}
Further, as $\RCD(\theta)$ lies in the trust region $\{\theta': KL(\theta'||\theta)\leq {\mathcal{E}'}\}$
we have that $KL\lb \RCD(\theta)||\theta\rb\leq{\mathcal{E}'}$, so we can write
\begin{align*}
 \Cgamma\lb \RCD(\theta)\rb+\alpha KL\lb \RCD(\theta)||\theta\rb 
&\leq  \Cgamma\lb \RCD(\theta)\rb+\alpha {\mathcal{E}'}  
\end{align*}
and thus
\begin{align*}
\Jalpha(\theta)
&\leq  \Cgamma\lb \RCD(\theta)\rb+\alpha {\mathcal{E}'} . 
\end{align*}
Next we minimize both sides of this inequality within the trust region \mbox{$\{\theta': KL(\theta'||\theta)\leq \mathcal{E}\}$}.
We use that 
\begin{align*}
\Jalpha\lb \RJE(\theta) \rb=&\min_{
\substack{
\theta' \\
\text{s.t.}~KL(\theta'||\theta)\leq \mathcal{E}
}
} \Jalpha(\theta')
\end{align*}
and get
\begin{align}
\label{eq:proofsmooth1}
\Jalpha\lb \RJE(\theta) \rb\leq& \min_{
\substack{
\theta' \\
\text{s.t.}~KL(\theta'||\theta)\leq \mathcal{E}
}
} \lb \Cgamma\lb \RCD(\theta')\rb+\alpha {\mathcal{E}'} \rb. 
\end{align}
Now we use Lemma~\ref{th:lemma1} and rewrite the left hand side of this inequality.
\begin{align*}
\Jalpha\lb\RJE(\theta)\rb &=\left.  \Cgamma\lb \RCD\lb \RJE(\theta)\rb\rb \right\rvert_{{\mathcal{E}'}=\mathcal{E}^*_\alpha(\theta)} +\alpha \mathcal{E}^*_\alpha(\theta) 
\end{align*}
with $\mathcal{E}^*_\alpha(\theta) \coloneqq \mathcal{E}_\alpha(\RJE(\theta))$.
Plugging this back to \eqref{eq:proofsmooth1} we get
\begin{align*}
\left.  \Cgamma\lb \RCD\lb \RJE(\theta)\rb\rb \right\rvert_{{\mathcal{E}'}=\mathcal{E}^*_\alpha(\theta)} +\alpha \mathcal{E}^*_\alpha(\theta) 
\leq& \min_{
\substack{
\theta' \\
\text{s.t.}~KL(\theta'||\theta)\leq \mathcal{E}
}
} \lb \Cgamma\lb \RCD(\theta')\rb+\alpha {\mathcal{E}'} \rb.
\end{align*}
As this inequality holds for any ${\mathcal{E}'}>0$ we can plug in $\mathcal{E}^*_\alpha(\theta)$ on the right hand side of this inequality and obtain
\begin{align*}
\left.  \Cgamma\lb \RCD\lb \RJE(\theta)\rb\rb \right\rvert_{{\mathcal{E}'}=\mathcal{E}^*_\alpha(\theta)} +\alpha \mathcal{E}^*_\alpha(\theta) 
\leq& \min_{
\substack{
\theta' \\
\text{s.t.}~KL(\theta'||\theta)\leq \mathcal{E}
}
} \left. \Cgamma\lb \RCD(\theta')\rb\right\rvert_{{\mathcal{E}'}=\mathcal{E}^*_\alpha(\theta)}+\alpha  \mathcal{E}^*_\alpha(\theta).
\end{align*}
We subtract $\alpha\mathcal{E}^*_\alpha(\theta)$ on both sides
\begin{align*}
\left.  \Cgamma\lb \RCD\lb \RJE(\theta)\rb\rb \right\rvert_{{\mathcal{E}'}=\mathcal{E}^*_\alpha(\theta)} 
\leq& \min_{
\substack{
\theta' \\
\text{s.t.}~KL(\theta'||\theta)\leq \mathcal{E}
}
} \left. \Cgamma\lb \RCD(\theta')\rb\right\rvert_{{\mathcal{E}'}=\mathcal{E}^*_\alpha(\theta)}.
\end{align*}
Using Eq.~\eqref{eq:ft} gives
\begin{align*}
\left.  \Cgamma\lb \RCD\lb \RJE(\theta)\rb\rb \right\rvert_{{\mathcal{E}'}=\mathcal{E}^*_\alpha(\theta)} 
\leq& \left. C^*_{\mathcal{E},{\mathcal{E}'}}(\theta) \right\rvert_{{\mathcal{E}'}=\mathcal{E}^*_\alpha(\theta)}, \\
\end{align*}
which concludes the proof.
\end{proof}

\subsection{Proof of Statement 2}
Here we show that $\mathcal{E'}=\mathcal{E}^*_\alpha(\theta)$ is a monotonically decreasing function of $\alpha$.
$\mathcal{E}^*_\alpha(\theta)$ is given by
\begin{align*}
 \mathcal{E}^*_\alpha(\theta) = \mathcal{E}_\alpha\lb\RJE(\theta)\rb = \left. KL(\theta^*_{\alpha,\theta'}||\theta')\right\rvert_{\theta'=R^{\Jalpha}_{\mathcal{E}}(\theta)}.
\end{align*}
We have 
\begin{align*}
\left.\lb \alpha KL(\theta^*_{\alpha,\theta'}||\theta') + \Cgamma\lb \theta^*_{\alpha,\theta'} \rb \rb\right\rvert_{\theta'=R^{\Jalpha}_{\mathcal{E}}(\theta)} =& \left.\lb\inf_{\theta''} \alpha KL(\theta''||\theta') + \Cgamma(\theta'') \rb\right\rvert_{\theta'=R^{\Jalpha}_{\mathcal{E}}(\theta)}\\
=& \min_{
\substack{
\theta' \\
\text{s.t.}~KL(\theta'||\theta)\leq \mathcal{E}
}
}\inf_{\theta''} \alpha KL(\theta''||\theta') + \Cgamma(\theta'').
\end{align*}
For convenience we introduce a shorthand notation for the minimizers
\begin{align*}
\theta_\alpha &\coloneqq \RJE(\theta) \\
\theta'_\alpha &\coloneqq \theta^*_{\alpha,\theta'}|_{\theta'=\RJE(\theta)}.
\end{align*}
We compare $\alpha_1\geq 0$ with $\mathcal{E}^*_{\alpha_1}(\theta) \coloneqq KL(\theta'_{\alpha_1}||\theta_{\alpha_1})$
and $\alpha_2\geq 0$ with \mbox{$\mathcal{E}^*_{\alpha_2}(\theta) \coloneqq KL(\theta'_{\alpha_2}||\theta_{\alpha_2})$} and assume that $\mathcal{E}^*_{\alpha_1}(\theta) < \mathcal{E}^*_{\alpha_2}(\theta)$.
We show that from this it follows that $\alpha_1 > \alpha_2$.
\begin{proof}
As $\theta'_{\alpha_1}$,$\theta_{\alpha_1}$ minimize $\alpha_1 KL(\theta'||\theta) + \Cgamma(\theta')$ we have
\begin{align*}
\alpha_1 KL(\theta'_{\alpha_1}||\theta_{\alpha_1}) + \Cgamma(\theta'_{\alpha_1}) &\leq \alpha_1 KL(\theta'_{\alpha_2}||\theta_{\alpha_2}) + \Cgamma(\theta'_{\alpha_2}) \\
\Rightarrow \alpha_1 {\mathcal{E}}_{\alpha_1}(\theta) + \Cgamma(\theta'_{\alpha_1}) &\leq \alpha_1 {\mathcal{E}}_{\alpha_2}(\theta) + \Cgamma(\theta'_{\alpha_2})
\end{align*}
and analogous for $\alpha_2$
\begin{align*}
\alpha_2 KL(\theta'_{\alpha_1}||\theta_{\alpha_1}) + \Cgamma(\theta'_{\alpha_1}) &\geq \alpha_2 KL(\theta'_{\alpha_2}||\theta_{\alpha_2}) + \Cgamma(\theta'_{\alpha_2})
\\
\Rightarrow \alpha_2 {\mathcal{E}}_{\alpha_1}(\theta) + \Cgamma(\theta'_{\alpha_1}) &\geq \alpha_2 {\mathcal{E}}_{\alpha_2}(\theta) + \Cgamma(\theta'_{\alpha_2})
\end{align*}
With ${\mathcal{E}}_{\alpha_1}(\theta) < {\mathcal{E}}_{\alpha_2}(\theta)$ we get
\begin{align*}
\alpha_1 \geq \frac{\Cgamma(\theta'_{\alpha_1})-\Cgamma(\theta'_{\alpha_2})}{{\mathcal{E}}_{\alpha_2}(\theta)-{\mathcal{E}}_{\alpha_1}(\theta)} \geq \alpha_2.
\end{align*}
\end{proof}
We showed that from ${\mathcal{E}}_{\alpha_1}(\theta) < {\mathcal{E}}_{\alpha_2}(\theta)$ it follows that $\alpha_1 \geq \alpha_2$ which proofs that $\mathcal{E}_\alpha(\theta)$ is monotonously decreasing in $\alpha$.

Direct updates are myopic and do not take into account successive steps and are thus suboptimal when more than one update is needed. 
Smoothed updates on the other hand, as we see on theorem \ref{th:smooth}, anticipate a subsequent step and minimize the cost that results from this this two-step update.
Hence smoothed updates favour a greater cost reduction in the future over maximal cost reduction in the current step.
The strength of this anticipatory effect depends on the smoothing strength, which is controlled by the smoothing parameter $\alpha$:
For large $\alpha$, smoothing is weak and the size $\mathcal{E}'$ of this anticipated second step becomes small. 
Fig.~\ref{fig0} B illustrates that for this case, when  $\mathcal{E}'$ becomes small, smoothed updates become more similar to direct updates.
In the limiting case $\alpha \ra \infty$ the difference between smoothed and direct updates vanishes completely, as  $\Jalpha(\theta) \ra \Cgamma(\theta)$ (see section \ref{sec:alphagamma}). 

We expect that also with multiple update steps due to this anticipatory effect, iterating smoothed updates leads to a faster decrease of the cost than iterating direct updates.
We will confirm this by numerical studies.
Furthermore, we expect that this accelerating effect of smoothing is stronger for smaller values of $\alpha$.
On the other hand, as we will discuss in the next section, for smaller values of $\alpha$ it is harder to accurately perform the smoothed updates.
Therefore we expect an optimal performance for an intermediate value of $\alpha$.
Based on this we build an algorithm in the next section that aims to accelerate policy optimization by cost function smoothing.

\section{Additional Theoretical Results for Section~\ref{sec:algorithm}} 

\subsection{Smoothed Updates for Small Update Steps $\mathcal{E}$}
\label{ap:TRPOTRPI}
We want to compute Eq.~\eqref{eq:unconstraintargminf} for small $\mathcal{E}$ which corresponds to large $\beta$. Assuming a smooth dependence of $p_{u_\theta}$ on $\theta$, bounding $KL(\theta||\theta_n)$ to a very small value allows us to do a Taylor expansion which we truncate at second order:
\begin{align}
&\argmin_{\theta'} \Jalpha(\theta') + \beta KL(\theta'||\theta_{n}) \approx\\
&\qquad
\approx \argmin_{\theta'}~ (\theta'-\theta_n)^T \nabla_{\theta'}\Jalpha(\theta') + \frac{1}{2}(\theta'-\theta_n)^T \lb H +\beta F \rb (\theta'-\theta_n) \\
&\qquad= \theta_n - \beta^{-1}F^{-1}\left.\nabla_{\theta'}\Jalpha(\theta')\right\rvert_{\theta'=\theta_n} + \mathcal{O}(\beta^{-2})
\label{eq:natgradlagrange}
\end{align}
with 
\begin{align*}
H &= \left.\nabla_{\theta'} \nabla_{\theta'}^T \Jalpha(\theta')\right\rvert_{\theta'=\theta_n}\\
F &= \left.\nabla_{\theta'} \nabla_{\theta'}^T KL(\theta'||\theta_n)\right\rvert_{\theta'=\theta_n}.
\end{align*}
See also \cite{martens2014new}.
We used that $\mathcal{E}\ll 1 \Leftrightarrow \beta \gg 1$. With this the Fisher information $F$ dominates over the Hessian $H$ and thus the Hessian does not appear anymore in the update equation.
This defines a natural gradient update with stepsize $\beta^{-1}$.

\subsection{Inversion of the Fisher matrix}
\label{ap:4}
We compute an approximation to the natural gradient $g_f = F^{-1}g$ by approximately solving the linear equation $Fg_f = g$ using truncated conjugate gradient.
With the normal gradient $g$ and the Fisher matrix $F = \nabla_\theta \nabla_\theta^T KL(p_{u_{\theta}}||p_{u_{\theta_n}})$ (see App.~\ref{ap:TRPOTRPI}).

We use an efficient way to compute the Fisher vector product $F\myvec{y}$ \cite{schulman2015trust} using an automated differentiation package:
First for each rollout $i$ and timepoint $t$ the symbolic expression for the gradient on the KL multiplied
by a vector $\myvec{y}$ is computed:
\begin{align*}
a_{i,t}(\theta_{n+1})&=\lb \nabla_{\theta_{n+1}}^T \log \frac{\pi_{\theta_n}(\myvec{a}_t^i|t,\myvec{x}^i_{t})}{\pi_{\theta_{n+1}}(\myvec{a}_t^i|t,\myvec{x}^i_{t})} \rb \myvec{y}.
\end{align*}

Then we take the second derivative on this scalar quantity, sum over all times and average over the samples. This gives then the Fisher vector
\begin{align*}
F\myvec{y} &= \frac{1}{N}\sum_{i=1}^N \sum_{0<t<T} \nabla_{\theta_{n+1}} a_{i,t}(\theta_{n+1}).
\end{align*}

For practical reasons, we reverse the arguments of the KL, since it is easier to estimate it from samples drawn from the first argument. For very small values, the $KL$ is approximately symmetric in its arguments. Also, the equality in~\eqref{eq:eqtrust} differs from~\cite{schulman2015trust}, which optimizes a value function within the trust region, e.g., $KL(\theta_n||\theta_{n+1})\leq\mathcal{E}$.

\subsection{Proof for equivalence of weight entropy and KL-divergence}
\label{ap:3}
We want to show that 
\begin{align*}
\lim_{N\ra \infty} \log N - H_N(w) =&\lim_{N\ra \infty} \log N +\sum_{i=1}^N w^i \log(w^i) \\
=& KL(p^*_{\alpha,\theta}||p_{u_\theta}).
\end{align*}
Where the samples $i$ are drawn from $p_{u_\theta}$ and the $w^i$ are given by
\begin{align*}
w^i = \frac{1}{{\sum_i^N \exp \lb -\frac{1}{\gamma+\alpha} S_{p_{u_\theta}}(\tau^i) \rb}}\exp \lb -\frac{1}{\gamma+\alpha} S_{p_{u_\theta}}(\tau^i) \rb,
\end{align*}

We get
\begin{align*}
&\lim_{N\ra \infty} \log N +\sum_{i=1}^N w^i \log(w^i) = \\
=& \lim_{N\ra \infty} \log N +\sum_{i=1}^N  \frac{1}{{\sum_i^N \exp \lb -\frac{1}{\gamma+\alpha} \Sgamma_{p_{u_\theta}}(\tau^i) \rb}}\exp \lb -\frac{1}{\gamma+\alpha} \Sgamma_{p_{u_\theta}}(\tau^i) \rb 
\cdot \\ &\qquad\qquad\qquad\qquad\qquad \cdot 
\log\lb  \frac{1}{{\sum_i^N \exp \lb -\frac{1}{\gamma+\alpha} \Sgamma_{p_{u_\theta}}(\tau^i) \rb}}\exp \lb -\frac{1}{\gamma+\alpha} \Sgamma_{p_{u_\theta}}(\tau^i) \rb \rb \\
=& \lim_{N\ra \infty} \log N +\frac{1}{N}\sum_{i=1}^N  \frac{1}{\frac{1}{N}{\sum_i^N \exp \lb -\frac{1}{\gamma+\alpha} \Sgamma_{p_{u_\theta}}(\tau^i) \rb}}\exp \lb -\frac{1}{\gamma+\alpha} \Sgamma_{p_{u_\theta}}(\tau^i) \rb 
\cdot \\ &\qquad\qquad\qquad\qquad\qquad \cdot 
\log\lb  \frac{\frac{1}{N}}{{\frac{1}{N}\sum_i^N \exp \lb -\frac{1}{\gamma+\alpha} \Sgamma_{p_{u_\theta}}(\tau^i) \rb}}\exp \lb -\frac{1}{\gamma+\alpha} \Sgamma_{p_{u_\theta}}(\tau^i) \rb \rb \\
=& \lim_{N\ra \infty} \frac{1}{N}\sum_{i=1}^N  \frac{1}{\frac{1}{N}{\sum_i^N \exp \lb -\frac{1}{\gamma+\alpha} \Sgamma_{p_{u_\theta}}(\tau^i) \rb}}\exp \lb -\frac{1}{\gamma+\alpha} \Sgamma_{p_{u_\theta}}(\tau^i) \rb 
\cdot \\ &\qquad\qquad\qquad\qquad\qquad \cdot 
\log\lb  \frac{1}{{\frac{1}{N}\sum_i^N \exp \lb -\frac{1}{\gamma+\alpha} \Sgamma_{p_{u_\theta}}(\tau^i) \rb}}\exp \lb -\frac{1}{\gamma+\alpha} \Sgamma_{p_{u_\theta}}(\tau^i) \rb \rb
\end{align*}
Now we replace in the limit $N\ra \infty$, $\frac{1}{N}\sum_i^N \ra \lla \rra_{p_{u_\theta}}$:
\begin{align*}
=& \lla  \frac{1}{\lla \exp \lb -\frac{1}{\gamma+\alpha} \Sgamma_{p_{u_\theta}}(\tau) \rb \rra_{p_{u_\theta}}}\exp \lb -\frac{1}{\gamma+\alpha} \Sgamma_{p_{u_\theta}}(\tau) \rb 
\right.\cdot\\ &\cdot\left. 
\log\lb  \frac{1}{\lla \exp \lb -\frac{1}{\gamma+\alpha} \Sgamma_{p_{u_\theta}}(\tau) \rb \rra_{p_{u_\theta}}}\exp \lb -\frac{1}{\gamma+\alpha} \Sgamma_{p_{u_\theta}}(\tau) \rb \rb \rra_{p_{u_\theta}}
\end{align*}
Using Eq.~\eqref{eq:annealingcurve} this gives
\begin{align*}
=& \lla \log\lb  \frac{1}{\lla \exp \lb -\frac{1}{\gamma+\alpha} \Sgamma_{p_{u_\theta}}(\tau) \rb \rra_{p_{u_\theta}}}\exp \lb -\frac{1}{\gamma+\alpha} \Sgamma_{p_{u_\theta}}(\tau) \rb \rb \rra_{p^*_{\alpha,\theta}} \\
=& \lla \log\lb  \frac{1}{\lla \exp \lb -\frac{1}{\gamma+\alpha} \Sgamma_{p_{u_\theta}}(\tau) \rb \rra_{p_{u_\theta}}}\exp \lb -\frac{1}{\gamma+\alpha} \Sgamma_{p_{u_\theta}}(\tau) \rb \frac{p_{u_\theta}(\tau)}{p_{u_\theta}(\tau)} \rb \rra_{p^*_{\alpha,\theta}} \\
=& \lla \log \frac{p^*_{\alpha,\theta}(\tau)}{p_{u_\theta}(\tau)} \rra_{p^*_{\alpha,\theta}} \\
=&KL(p^*_{\alpha,\theta}||p_{u_\theta}).
\end{align*}

\subsection{The Smoothness Parameter ${\Delta}$ is monotonic in $\alpha$}
\label{ap:deltaalpha}
Now we show that
\begin{align*}
{\Delta} = KL(p^*_{\alpha,\theta}||p_{u_{\theta}})
\end{align*}
is a monotonic function of $\alpha$.
\begin{align*}
\frac{\partial}{\partial \alpha} KL(p^*_{\alpha,\theta}||p_{u_{\theta}})&= \frac{\partial}{\partial \alpha}\lla \ln \frac{p^*_{\alpha,\theta}}{p_{u_{\theta}}} \rra_{p^*_{\alpha,\theta}} \\ 
&= \frac{\partial}{\partial \alpha}\lla \frac{p^*_{\alpha,\theta}}{p_{u_{\theta}}} \ln \frac{p^*_{\alpha,\theta}}{p_{u_{\theta}}} \rra_{p_{u_{\theta}}} \\
&=\lla \lb \frac{\partial}{\partial \alpha}  \frac{p^*_{\alpha,\theta}}{p_{u_{\theta}}} \rb\ln \frac{p^*_{\alpha,\theta}}{p_{u_{\theta}}} \rra_{p_{u_{\theta}}}  
+
\lla \frac{p^*_{\alpha,\theta}}{p_{u_{\theta}}} \frac{\partial}{\partial \alpha} \ln \frac{p^*_{\alpha,\theta}}{p_{u_{\theta}}} \rra_{p_{u_{\theta}}} \\
&=\lla \lb \frac{\partial}{\partial \alpha}  \frac{p^*_{\alpha,\theta}}{p_{u_{\theta}}} \rb\ln \frac{p^*_{\alpha,\theta}}{p_{u_{\theta}}} \rra_{p_{u_{\theta}}}  
+
\lla \frac{1}{p_{u_{\theta}}} \frac{\partial}{\partial \alpha} p^*_{\alpha,\theta} \rra_{p_{u_{\theta}}} \\
&=\lla \lb \frac{\partial}{\partial \alpha}  \frac{p^*_{\alpha,\theta}}{p_{u_{\theta}}} \rb\ln \frac{p^*_{\alpha,\theta}}{p_{u_{\theta}}} \rra_{p_{u_{\theta}}}  
+
\frac{\partial}{\partial \alpha}\lla  1\rra_{p^*_{\alpha,\theta}} \\
&=\lla \lb \frac{\partial}{\partial \alpha}  \frac{p^*_{\alpha,\theta}}{p_{u_{\theta}}} \rb\ln \frac{p^*_{\alpha,\theta}}{p_{u_{\theta}}} \rra_{p_{u_{\theta}}}.  \\
\end{align*}
Now let us look at 
\begin{align*}
 \frac{\partial}{\partial \alpha}  \frac{p^*_{\alpha,\theta}}{p_{u_{\theta}}} 
 &= \frac{\partial}{\partial \alpha}  \lb \frac{1}{Z^\alpha_{p_{u_{\theta}}}} \exp \lb -\frac{1}{\gamma+\alpha} \Sgamma_{p_{u_{\theta}}}(\tau) \rb \rb
\\
Z^\alpha_{p_{u_{\theta}}}&=\lla \exp \lb -\frac{1}{\gamma+\alpha} \Sgamma_{p_{u_{\theta}}}(\tau) \rb \rra_{p_{u_{\theta}}}.
 \end{align*}
we get
\begin{align*}
\frac{\partial}{\partial \alpha}  \frac{p^*_{\alpha,\theta}}{p_{u_{\theta}}}  &= \frac{1}{\lb \gamma + \alpha  \rb^2} \Sgamma_{p_{u_{\theta}}}(\tau)\frac{p^*_{\alpha,\theta}}{p_{u_{\theta}}} - \frac{p^*_{\alpha,\theta}}{p_{u_{\theta}}} \frac{1}{Z^\alpha_{p_{u_{\theta}}}} \frac{\partial}{\partial \alpha} Z^\alpha_{p_{u_{\theta}}} \\
\frac{\partial}{\partial \alpha} Z^\alpha_{p_{u_{\theta}}}&=\lla \frac{1}{\lb \gamma + \alpha  \rb^2} \Sgamma_{p_{u_{\theta}}} \exp \lb -\frac{1}{\gamma+\alpha} \Sgamma_{p_{u_{\theta}}}(\tau) \rb \rra_{p_{u_{\theta}}}.
\end{align*}
and thus 
\begin{align*}
\frac{\partial}{\partial \alpha}  \frac{p^*_{\alpha,\theta}}{p_{u_{\theta}}}  &= 
\frac{1}{\lb \gamma + \alpha  \rb^2} \Sgamma_{p_{u_{\theta}}}(\tau)\frac{p^*_{\alpha,\theta}}{p_{u_{\theta}}} - \frac{p^*_{\alpha,\theta}}{p_{u_{\theta}}} \frac{1}{\lb \gamma + \alpha  \rb^2} \lla \Sgamma_{p_{u_{\theta}}} \rra_{p^*_{\alpha,\theta}} \\
&= 
\frac{1}{\lb \gamma + \alpha  \rb^2}\frac{p^*_{\alpha,\theta}}{p_{u_{\theta}}} \lb \Sgamma_{p_{u_{\theta}}}(\tau)-  \lla \Sgamma_{p_{u_{\theta}}} \rra_{p^*_{\alpha,\theta}} \rb. \\
\end{align*}
So finally we get
\begin{align*}
\frac{\partial}{\partial \alpha} KL(p^*_{\alpha,\theta}||p_{u_{\theta}}) 
&= \frac{1}{\lb \gamma + \alpha  \rb^2} \lla \frac{p^*_{\alpha,\theta}}{p_{u_{\theta}}} \lb \Sgamma_{p_{u_{\theta}}}(\tau)-  \lla \Sgamma_{p_{u_{\theta}}} \rra_{p^*_{\alpha,\theta}} \rb \ln \frac{p^*_{\alpha,\theta}}{p_{u_{\theta}}} \rra_{p_{u_{\theta}}}  \\
&= \frac{1}{\lb \gamma + \alpha  \rb^2} \lla \frac{p^*_{\alpha,\theta}}{p_{u_{\theta}}} \lb \Sgamma_{p_{u_{\theta}}}(\tau)-  \lla \Sgamma_{p_{u_{\theta}}} \rra_{p^*_{\alpha,\theta}} \rb \lb  -\frac{1}{\gamma+\alpha} \Sgamma_{p_{u_{\theta}}}(\tau)  -\log Z^\alpha_{p_{u_{\theta}}} \rb  \rra_{p_{u_{\theta}}}  \\
&= \frac{1}{\lb \gamma + \alpha  \rb^2} \lla \lb \Sgamma_{p_{u_{\theta}}}(\tau)-  \lla \Sgamma_{p_{u_{\theta}}} \rra_{p^*_{\alpha,\theta}} \rb \lb  -\frac{1}{\gamma+\alpha} \Sgamma_{p_{u_{\theta}}}(\tau)  -\log Z^\alpha_{p_{u_{\theta}}} \rb  \rra_{p^*_{\alpha,\theta}}  \\
&= -\frac{1}{\lb \gamma + \alpha  \rb^3} \lb \lla \lb  \Sgamma_{p_{u_{\theta}}}\rb^2 \rra_{p^*_{\alpha,\theta}}-  \lla \Sgamma_{p_{u_{\theta}}} \rra_{p^*_{\alpha,\theta}}^2  \rb \\
&= -\frac{1}{\lb \gamma + \alpha  \rb^3} \text{Var}\lb \Sgamma_{p_{u_{\theta}}} \rb \leq 0.
\end{align*}
Therefore
\begin{align*}
{\Delta} = KL(p^*_{\alpha,\theta}||p_{u_{\theta}})
\end{align*}
is a monotonically decreasing function of $\alpha$.

\begin{algorithm}[t]
\caption{ASPIC - Adaptive Smoothing of Path Integral Control}\label{alg:alg1}
\begin{algorithmic}
\REQUIRE State cost function $V(\myvec{x},t)$\\
$\qquad \qquad$ control cost parameter $\mybar{\gamma}$ \\
$\qquad \qquad$ base policy that defines uncontrolled dynamics $\pi_0$ \\
$\qquad \qquad$ simulator of system dynamics with a parametrized policy $\pi_\theta$ \\
$\qquad \qquad$ trust region sizes $\mathcal{E}$ \\
$\qquad \qquad$ smoothing strength $\Delta$ \\
$\qquad \qquad$ number of samples $N$
\STATE initialize $\theta_0$
\STATE $n=0$
\REPEAT 
\STATE draw samples $\tau^i$, with $i=1,\hdots,N$, from simulator controlled by parametrized policy $\pi_{\theta_n}$
\STATE for each sample $i$  compute $\Sgamma_{p_{u_{\theta_n}}}(\tau^i)=\sum_{0<t<T} V(\myvec{x}^i_t,t)+\mybar{\gamma}\log \frac{\pi_{{\theta_n}}(\myvec{a}_t^i|t,\myvec{x}^i_{t})}{\pi_0(\myvec{a}_t^i|t,\myvec{x}^i_{t})}$ 
\STATE \COMMENT{Find minimal $\alpha$ such that $KL\leq \Delta$}
\STATE $\alpha \leftarrow 0$
\REPEAT 
\STATE increase $\alpha$
\STATE $S^i_\alpha \leftarrow \Sgamma_{p_{u_{\theta_n}}}(\tau^i) \cdot \frac{1}{\mybar{\gamma}+\alpha}$
\STATE compute weights $w_i \leftarrow \exp(-S^i_\alpha)$
\STATE normalize weights $w_i \leftarrow \frac{w_i}{\sum_i(w_i)}$
\STATE compute sample size independent weight entropy $KL \leftarrow \log{N}+\sum_i w_i\log(w_i)$ 
\UNTIL{$KL\leq \Delta$}
\STATE \COMMENT{whiten the weigths}
\STATE $\hat{w}_i \leftarrow \frac{w_i-\text{mean}(w_i)}{\text{std}(w_i)}$
\STATE \COMMENT{compute the gradient on the smoothed cost}
\STATE $g \leftarrow \sum_i \sum_t \hat{w}_i \left. \frac{\partial}{\partial \theta} \log \pi_\theta(\myvec{a}_t^i|t,\myvec{x}^i_{t})\right\rvert_{\theta=\theta_n}$
\STATE \COMMENT{compute Fisher matrix}
\STATE use conjugate gradient descent to compute an approximate solution to the natural gradient $g_F=F^{-1}g$ (see App.~\ref{ap:4})
\STATE do line search to compute step size $\eta$ such $KL(\theta_n||\theta_{n+1})=\mathcal{E}.$
\STATE update parameters $\theta_{n+1} \leftarrow \theta_n + \eta \cdot g_F$ 
\STATE $n=n+1$
\UNTIL{convergence}
\end{algorithmic}
\end{algorithm}

\section{Experimental Details and Additional Results}
\label{ap:addexp}
Algorithm~\ref{alg:alg1} summarizes ASPIC.
We first analyze the behavior of ASPIC in a simple linear-quadratic control problem,~\ref{sec:LQ},\ref{ap:fpLQ}.
We then look at the dependence on the number of rollouts per iteration~$N$ in~\ref{ap:expN} and the interplay between smoothing strength $\Delta$ and trust region size $\mathcal{E}$ in~\ref{ap:expinter}. Finally, we describe the parameter settings for all tasks in~\ref{ap:detnexp}.

\begin{figure}[t]
\vskip 0.2in
\begin{center}
\centerline{\includegraphics[width=.5\columnwidth]{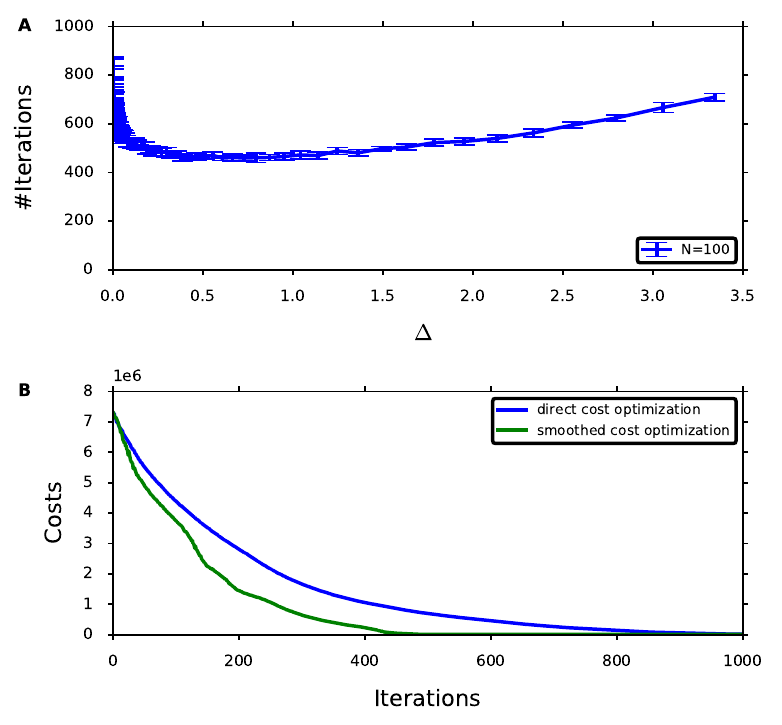}}
\caption{LQ control problem: Brownian viapoints. For each iteration we used $N=100$ rollouts to compute the gradient. A) Number of iterations needed for the cost to cross a threshold $\Cgamma\leq2\cdot 10^4$ versus the smoothing strength $\Delta$.
For $\Delta=0$ there is no smoothing. Increasing the smoothing strength results in a faster decrease of the cost; when $\Delta$ is increased further the performance decreases again.
Errorbars denote mean and standard deviation over $10$ runs of the algorithm.
B) Cost versus the iterations of the algorithm. Direct optimization of the cost exhibits a slower convergence rate than optimization of the smoothed cost with $\Delta=0.2\log100$.
}
\label{fig1}
\end{center}
\vskip -0.2in
\end{figure} 

\subsection{A Simple Linear-Quadratic Control Problem: Brownian Viapoints}
\label{sec:LQ}
We analyse the convergence speed for different values of the smoothing strength $\Delta$ in the task of controlling a one-dimensional Brownian particle
\begin{align}
\dot{x} = u(x,t) + \xi.
\label{eq:brownianparticle}
\end{align}
We define the state cost as a quadratic penalty for deviating from the viapoints $x_i$ at the different times $t_i$: $V(x,t) = \sum_{i} \delta\lb t-t_i\rb \frac{\lb x-x_i \rb^2}{2\sigma^2}$ with $\sigma = 0.1$. As a parametrized controller we use a time varying linear feedback controller, i.e., $u_\theta(x,t) =  \theta_{1,t}x+\theta_{0,t}$. This controller fulfils the requirement of full parametrization for this task (see App.~\ref{ap:fpLQ}).
For further details of the numerical experiment see appendix~\ref{ap:detnexp}.

We apply ASPIC to this control problem and compare its performance for different sizes of the smoothing strength $\Delta$ (see Fig.~\ref{fig1}). 
The results confirm our expectations from our theoretical analysis. 
As predicted by theory we observe an acceleration of the policy optimization when smoothing is switched on.
This acceleration becomes more pronounced when $\Delta$ is increased, which we attribute to an increase of the anticipatory effect of the smoothed updates as smoothing becomes stronger (see section \ref{sec:Theory}).
When $\Delta$ is too large the performance of the algorithm deteriorates again, which is in line with our discussion of gradient estimation problems that arise for strong smoothing. 

\subsection{Full parametrization in LQ problem}
\label{ap:fpLQ}
Here we discuss why for a linear quadratic problem a time varying linear controller is a full parametrization.
We want to show that for every
\begin{align}
p^*_{\alpha,{\theta_0}}&=\frac{1}{Z} p_{u_0}(\tau) \exp \lb -\frac{1}{\gamma+\alpha} \Sgamma_{p_{u_{\theta_0}}}(\tau) \rb
\end{align}
there is a time varying linear controller $u_{\theta_{\alpha,{\theta_0}}^*}$ such that $p_{u_{\theta_{\alpha,{\theta_0}}^*}}=p^*_{\alpha,{\theta_0}}$.
We  assume that $u_{\theta_0}$ is a time varying linear controller.
In App.~\ref{ap:5} we have shown that $u^*_{\alpha,{\theta_0}}$ is the solution to the Path Integral control problem with dynamics
\begin{align*}
\dot{\myvec{x}}_t=\myvec{f}(\myvec{x}_t,t) + g(\myvec{x}_t,t)\left( \myvec{\tilde{u}}(\myvec{x}_t,t)+\myvec{\hat{u}}(\myvec{x}_t,t) + \myvec{\xi}_t \right)
\end{align*}
and cost
\begin{align}
\lla \int_0^T \frac{1}{\gamma}V(\myvec{x}_t,t) -\frac{1}{2} \frac{\mybar{\gamma}}{\alpha}\myvec{\tilde{u}}(\myvec{x}_t,t)^T \myvec{\tilde{u}}(\myvec{x}_t,t)  dt+\int_0^T \lb\frac{1}{2} \myvec{\hat{u}}(\myvec{x}_t,t)^T \myvec{\hat{u}}(\myvec{x}_t,t) + \myvec{\hat{u}}(\myvec{x}_t,t)^T\myvec{\xi}_t\rb dt \rra_{p_{\hat{u}}},
\end{align}
with $\myvec{\tilde{u}}=\lb 1-\frac{\mybar{\gamma}}{\gamma+\alpha} \rb  \myvec{u}_{\theta_0}(\myvec{x}_t,t)$.

It is now easy to see that if $u_{\theta_0}$ is a time varying linear controller, thus a linear function of the state,  the cost is a quadratic function of the state $x$ (note that $V(\myvec{x}_t,t)$ is quadratic in the LQ case).
Thus for all values of $\alpha$, $u^*_{\alpha,{\theta_0}}$ is the solution to a linear quadratic control problem and thus a time varying linear controller (see e.g. \cite{kwakernaak1972linear}). Therefore a time varying linear controller is a full parametrization.

\subsection{Dependence on the Number of Rollouts per Iteration $N$}
\label{ap:expN}

\begin{figure}[t]
\vskip 0.2in
\begin{center}
\includegraphics[width=.5\columnwidth]{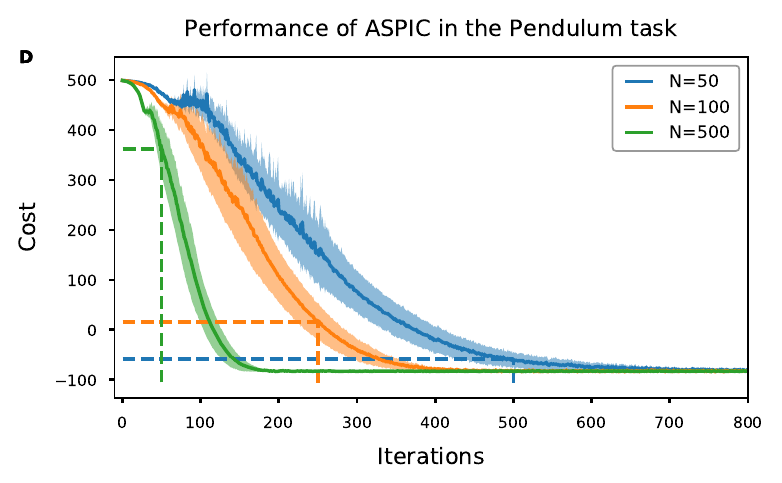}
\caption{
Performance as a function of the number of iterations for different values of $N\in\{50,100,500\}$ in the Pendulum swing-up task.
Dashed lines denote the solution for a total fixed budget of $25$K rollouts, i.e., $500$, $250$, and $50$ iterations, respectively.
In this case, $N=50$ achieves near optimal performance whereas using larger values of $N$ leads to worse solutions.
}
\label{fig:exp3}
\end{center}
\vskip -0.2in
\end{figure} 

We now analyze the dependence of the performance of ASPIC on the number of rollouts per iteration~$N$. In general, using larger values of $N$ allows for more reliable gradient estimates and achieves convergence in fewer iterations. However, too large $N$ may be inefficient and lead to suboptimal solutions in the presence of a fixed budget of rollouts.

Figure~\ref{fig:exp3} illustrates this trade-off in the Pendulum swing-up task for three values of $N$. For a total budget of $25$K rollouts (dashed lines), the lowest value of $N=50$ achieves near optimal performance and is preferable to the other choices, despite resulting in higher variance estimates and requiring more iterations until convergence. In particular, the solutions achieved using $N=500$ have cost~$>350$, while for $N=50$, all solutions have cost $< -50$.

\subsection{Interplay Between Smoothing Strength $\Delta$ and Trust Region Size $\mathcal{E}$}
\label{ap:expinter}
To understand better the relation between the smoothing strength and the trust region sizes, we analyze empirically the performance of ASPIC as a function of both $\Delta$ and $\mathcal{E}$ parameters.
We focus on the Acrobot task and in the setting of $N=500$ and intermediate smoothing strength, when smoothing is most beneficial. 

Figure~\ref{fig:deltaeps} shows the cost as a function of $\Delta$ and $\mathcal{E}$ averaged over the first $500$ iterations of the algorithm, and for $10$ different runs.
Larger (averaged) costs correspond runs where the algorithm fails to converge. Conversely, the lower cost, the fastest the convergence.
In general, larger values of $\mathcal{E}$ lead to faster convergence.
However, the convergence is less stable for smaller values of $\Delta$.
For stronger smoothing, the algorithm is more sensitive to $\mathcal{E}$.

\begin{figure}[t]
\vskip 0.2in
\begin{center}
\centerline{\includegraphics[width=.5\columnwidth]{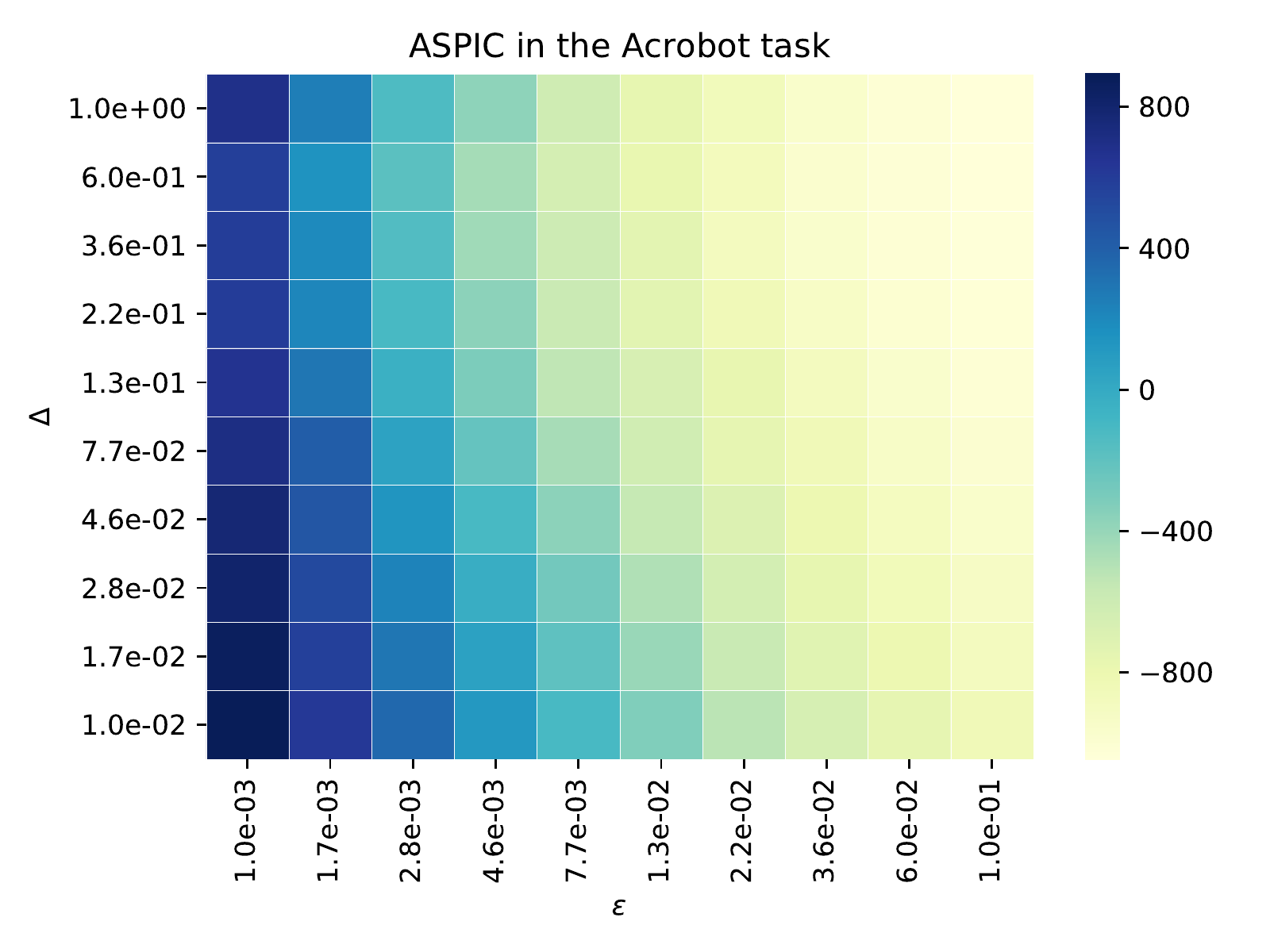}}
\caption{
Solution cost as a function of the smoothing strength $\Delta$ and the trust region size $\varepsilon$ in the Acrobot task.
Shown is the cost averaged over the first $500$ iterations of the algorithm, and for $10$ different runs.
Blue indicates failure to convergence. White indicates the solutions which converged fastest.
}
\label{fig:deltaeps}
\end{center}
\vskip -0.2in
\end{figure}

\subsection{Details of Numerical Experiments}
\label{ap:detnexp}

\subsection*{\underline{Linear-Quadratic control Task}}
\begin{description}
\item[Dynamics:] The dynamics are ODEs integrated by an Euler scheme (see section \ref{sec:LQ}). The differential equation is initialized at $x=0$. $dt=0.1$
\item[Control problem:] $\gamma = 1$. Time-Horizon $T=10s$. State-Cost function: see section~\ref{sec:LQ}. $(x_0,t_0)=(-10,1)$, $(x_1,t_1)=(10,2)$,$(x_2,t_2)=(-10,3)$, $(x_3,t_3)=(-20,4)$, $(x_4,t_4)=(-100,5)$, $(x_5,t_5)=(-50,6)$, $(x_6,t_6)=(10,7)$, $(x_7,t_7)=(20,8)$, $(x_8,t_8)=(30,9)$. Variance of uncontrolled dynamics  $\nu=1$.
\item[Algorithm:] Batchsize: $N=100$. $\mathcal{E}=0.1$. $\Delta = 0.2\log100$. Conjugate gradient iterations: 2 (for each time step separately). The parametrized controller was initialized at $\theta=0$.
\end{description}

\subsection*{\underline{Pendulum Task}}
\begin{description}
\item[Dynamics:] The differential equation for the pendulum is:
\begin{align*}
\ddot{x} + c\omega_0\dot{x} +\omega_0^2\sin(x) = \lambda \lb u + \xi \rb
\end{align*}
with
\begin{itemize}
\item $c\omega_0=0.1$ [$s^{-1}$]
\item $\omega_0^2 = 10.$ [$s^{-2}$]
\item $\lambda=0.2$
\end{itemize}
We implemented this differential equation as a first order differential equation and integrated it with an Euler scheme with $dt=0.01$.
The pendulum is initialized resting at the bottom: 
\begin{align*}
\dot{x}=0,x=0.
\end{align*}
As a parametrized controller we use a time varying linear feedback controller: 
\begin{align*}
u_\theta(x,\dot{x},t) =  
\theta_{3,t}\cos(x)+
\theta_{2,t}\sin(x)+
\theta_{1,t}\dot{x}+
\theta_{0,t}.
\end{align*}
The parametrized controller was initialized at $\theta=0$.
\item[Control-problem:] $\gamma = 1.$. $T=3.0s$. The State-Cost function has End-Cost only: 
\begin{align*}
V(x,\dot{x},t)=\delta(t-T) \lb-500Y+10\dot{x}^2 \rb
\end{align*}
 with $Y= -\cos(x)$ (height of tip). Variance of uncontrolled dynamics  $\nu=1$

\item[Algorithm:] Batchsize: $N=500$. $\mathcal{E}=0.1$. $\Delta = 0.5$. The Fisher-matrix was inverted for each time step separately using the scipy pseudo-inverse with \mbox{rcond=1e-4}.
\end{description}

\subsection*{\underline{Acrobot Task}}
\begin{description}
\item[Dynamics:] We use the definition of the acrobot as in \cite{spong1995swing}.
The differential equations for the acrobot are:
\begin{align*}
d_{11}(x)\ddot{x}_1 + d_{12}(x)\ddot{x}_2 + h_1(x,\dot{x}) + \phi_1(x) &= 0 \\
d_{21}(x)\ddot{x}_1 + d_{22}\ddot{x}_2 + h_2(x,\dot{x}) + \phi_2(x) &= \lambda\cdot\lb u + \xi \rb
\end{align*}
with
\begin{align*}
d_{11} &= m_1  l_{c1}^2 + m_2 \lb l_1^2 + l_{c2}^2 + 2 l_1 l_{c2} \cos(x_2 ) \rb + I_1 + I_2 \\      
d_{12} &= m_2 \lb l_{c2}^2 + l_1 l_{c2} \cos(x_2) \rb + I_2 \\     
d_{21} &= d_{12} \\
d_{22} &= m_2 l_{c2}^2 + I_2 \\
h_1 &= -m_2 l_1 l_{c2} \sin (x_2) \lb \dot{x}_2^2 + 2\dot{x}_1\dot{x}_2 \rb\\
h_2 &= m_2 l_1 l_{c2} \sin (x_2) \dot{x}_1^2 \\
\phi_2 &= m_2 l_{c2} G \cos\lb x_1 + x_2\rb \\     
\phi_1 &= (m_1 l_{c1} + m_2 l_1) g \cos\lb x_1\rb + \phi_2 \\
\end{align*}
with the parameter values
\begin{itemize}
\item $G=9.8$
\item $l_1 = 1.$ [m]
\item $l_2 = 2.$ [m]
\item $m_1 = 1.$ [kg] mass of link 1
\item $m_2 = 1.$ [kg] mass of link 2
\item $l_{c1} = 0.5$ [m] position of the center of mass of link 1
\item $l_{c2} = 1.0$ [m] position of the center of mass of link 2
\item $I_1 = 0.083$ moments of inertia for both links
\item $I_2 = 0.33$ moments of inertia for both links
\item $\lambda = 0.2$
\end{itemize}
We implemented this differential equation as a first order differential equation and integrated it with an Euler scheme with $dt=0.01$.
The acrobot is initialized resting at the bottom: 
\begin{align*}
\dot{x}_1=0,\dot{x}_2=0,x_1=-\frac{1}{2}\pi, x_2 = 0.
\end{align*}
As a parametrized controller we use a time varying linear feedback controller: 
\begin{align*}
u_\theta(x,\dot{x},t) = & 
\theta_{8,t}\cos(x_1)+
\theta_{7,t}\sin(x_2)+
\theta_{6,t}\cos(x_2)+
\theta_{5,t}\sin(x_2)+
\\ & +
\theta_{4,t}\sin(x_1+x_2)+
\theta_{3,t}\cos(x_1+x_2)+
\theta_{2,t}\dot{x_1}+
\theta_{1,t}\dot{x_2}+
\theta_{0,t}.
\end{align*}
The parametrized controller was initialized at $\theta=0$.
\item[Control-problem:] $\gamma = 1.$. Time-Horizon: $T=3.0s$. The State-Cost function has End-Cost only: 
\begin{align*}
V(x,\dot{x},t)=\delta(t-T) \lb-500Y+10(\dot{x_1}^2+\dot{x_2}^2)\rb
\end{align*}
 with $Y= -l_1 \cos(x_1) - l_2 \cos(x_1 + x_2)$ (height of tip).
Variance of uncontrolled dynamics  $\nu=1$.

\item[Algorithm:] Batchsize: $N=500$. $\mathcal{E}=0.1$. $\Delta = 0.5$. The Fisher-matrix was inverted for each time step separately using the scipy pseudo-inverse with \mbox{rcond=1e-4}.
\end{description}

\subsection*{\underline{Walker}}
\begin{description}
\item[Dynamics:] For dynamics and the state cost function we used "BipedalWalker-v2" from the OpenAI gym \cite{1606.01540}.
The policy was a Gaussian policy, with static variance $\sigma=1$.  The state dependent mean of the Gaussian policy  was a neural network controller with two hidden layers with 32 neurons, each. The activation function is a tanh. For the initialization we used Glorot Uniform  (see \cite{glorot2010understanding}).
The inputs to the neural network was the observation space provided by OpenAI gym task "BipedalWalker-v2":
State consists of hull angle speed, angular velocity, horizontal speed, vertical speed, position of joints and joints angular speed, legs contact with ground, and 10 lidar rangefinder measurements.
			
\item[Control-problem:] $\gamma = 0$. Time-Horizon: defined by OpenAI gym task ``BipedalWalker-v2''.
State-Cost function defined by OpenAI gym task "BipedalWalker-v2": Reward is given for moving forward, total 300+ points up to the far end. If the robot falls, it gets -100. Applying motor torque costs a small amount of points, more optimal agent will get better score.

\item[Algorithm:] Batchsize: $N=100$. $\mathcal{E}=0.01$. $\Delta = 0.05 \log 100$. Conjugate gradient iterations: 10.
\end{description}


\end{appendix}

\end{document}